%% file: main.tex
\documentclass[11pt]{article}
\usepackage{fullpage}
\usepackage{graphicx}
\usepackage{tikzsymbols}
\usepackage{amsmath}
\usepackage{amsthm}
\usepackage{amssymb}
\definecolor{Darkblue}{rgb}{0,0,0.4}
\definecolor{Brown}{cmyk}{0,0.61,1.,0.60}
\definecolor{Purple}{cmyk}{0.45,0.86,0,0}
\definecolor{Darkgreen}{rgb}{0.133,0.543,0.133}
\usepackage[colorlinks,linkcolor=Darkblue,filecolor=blue,citecolor=blue,urlcolor=Darkblue,pagebackref]{hyperref}
\usepackage[nameinlink]{cleveref}
\usepackage{aliascnt}
\usepackage{thm-restate}
\usepackage{appendix}
\usepackage{paralist}
\usepackage{multicol}
\usepackage{lineno}

\theoremstyle{plain}
\newtheorem{theorem}{Theorem}[section]
\newtheorem{claim}{Claim}

\newtheorem{lemma}[theorem]{Lemma}
\newtheorem{observation}[theorem]{Observation}

\newtheorem{definition}{Definition}

\newtheorem{remark}{Remark}

\input{mymacros-disc}

\title{How to Walk a Dog in Parallel: On Parallel Computation of the Discrete Fr\'echet Distance} 

\author{%
  Leonid Barenboim\thanks{%
    The Open University of Israel.
    Email: \texttt{leonidb@openu.ac.il}.}
  \and
  Arnold Filtser\thanks{%
    Bar-Ilan University.
    Email: \texttt{arnold.filtser@biu.ac.il}.}
  \and
  Omrit Filtser\thanks{%
    The Open University of Israel.
    Email: \texttt{omrit.filtser@gmail.com}.}
  \and
  Orr Fischer\thanks{%
    Bar-Ilan University.
    Email: \texttt{orr.fischer@biu.ac.il}.
}}

\begin{document}

\maketitle

\begin{abstract}
The Fr\'echet distance (also known as the dog-leash distance) is a well-studied and prevalent distance measure for polygonal curves. In this paper we present parallel algorithms for computing the discrete \frechet\ distance between two curves, in both the PRAM and MPC models. 

For two polygonal curves of length $n$, we show that the distance between them can be computed in the PRAM model in $O(\log n/ \sigma)$ rounds and $\tilde{O}(n^{2+\sigma})$ total work for any $0<\sigma<1$. In the sublinear MPC model with $n^\rho$ space per machine, it can be computed in $O(\log n / \rho^2)$ rounds and $\tilde{O}(n^{2+\rho})$ total space.
For $c$-packed curves, we match our bounds to the sequential algorithms, showing that a $(1+\eps)$-approximation of the distance can be computed in the PRAM model in $O(\log{n}+1/\sigma^3)$ rounds and $\tilde{O}(n^{1+\sigma})$ total work, and in the sublinear MPC model in $O(1 /\rho^3)$ rounds and $\tilde{O}(n^{1+\rho})$ total space, for any constant $\eps>0$. 
As part of our algorithms for $c$-packed curves, we also present parallel algorithms for computing curve simplifications, which may be of independent interest.
\end{abstract}

	\vfill
	\begin{multicols}{2}
		{
   \small 
			\setcounter{secnumdepth}{5}
			\setcounter{tocdepth}{2} \tableofcontents
		}
	\end{multicols}

    \thispagestyle{empty}
	\newpage
	\pagenumbering{arabic}

\thispagestyle{empty}
\clearpage
\setcounter{page}{1}

\section{Introduction}
The \emph{(discrete) \frechet\ distance} is a prevalent method for comparing curves. Polygonal curves are commonly employed to model data across diverse domains, including cartography, surgical planning, epidemiology, image processing, biology, and more. This widespread use has created a strong demand for efficient computation of distance measures that align with our intuitive understanding of curve similarity. Consequently, the \frechet\ distance (and its discrete variant) has attracted significant attention over the past three decades, following the seminal works of Eiter and Mannila~\cite{EM94} and Alt and Godau~\cite{AG95}.

The (continuous) \frechet\ distance is often described (informally) by the following analogy. Consider a person and a dog connected by a leash, each walking along a curve from its starting point to its end point. Both can control their speed but they are not allowed to backtrack. The \frechet\ distance between the two curves is the minimum length of a leash that is sufficient for traversing both curves in this manner.
The \emph{discrete \frechet\ distance} (DFD for short) is a popular variant in which the two curves are replaced by two sequences of points, and the person and dog are replaced by two frogs, hopping from one point to the next along their respective sequences. The discrete \frechet\ distance is defined as the smallest maximum distance between the frogs that can be achieved in such a joint sequence of hops. It is considered a good approximation of the continuous distance (when the sequences of points are dense enough samples), and is somewhat easier to compute. 

For two polygonal curves, each of length $n$, Alt and Godau~\cite{AG95} showed that the \frechet\ distance can be computed in $O(n^2\log n)$ time. Eiter and Mannila~\cite{EM94} showed that the discrete \frechet\ distance can be computed in $O(n^2)$ time, using a simple dynamic programming algorithm. However, in applications where input curves are extremely large, the quadratic running time is considered infeasible. For many years it has been an open problem whether it is possible to compute the continuous or discrete \frechet\ distance in subquadratic time. Some logarithmic-factor improvements were shown~\cite{BBMM14,AAKS14}, until a sequence of papers \cite{Bringmann14, BM16, BOS19} showed that strongly subquadratic algorithms are unlikely to exist, even in 1D and even if the solution may be approximated up to a factor of 3. A sequence of papers considered approximation algorithms for both discrete and continuous \frechet\ distance~\cite{BM16,CR18,CF21,HKOS22}, and only very recently Cheng, Huang, and Zhang~\cite{CHZ25} managed to achieve a constant factor approximation algorithm in strongly subquadratic time.
\old{Bringmann~\cite{Bringmann14}, and later Bringmann and Mulzer~\cite{BM16}, presented a SETH-based conditional lower bound implying that strongly subquadratic algorithms for the (discrete and continuous) \frechet\ distance are unlikely to exist, even in 1D and even if the solution may be approximated up to a factor of 1.399. Buchin, Ophelders, and Speckmann~\cite{BOS19} strengthen this lower bound to a $3$-approximation.

As a result, approximating the \frechet\ distance in near-linear time became an even more intriguing and fundamental question that fascinated many researchers in theory and practice. There was a sequence of papers considering approximation algorithms for both discrete and continuous \frechet\ distance~\cite{BM16,CR18,CF21,HKOS22}. However, no constant approximation in strongly subquadratic time was known until a very recent breakthrough by Cheng, Huang, and Zhang~\cite{CHZ25}, who presented a randomized algorithm in $O(n^{1.99}\log(n/\eps))$ time that computes a $(7+\eps)$-approximation for both the discrete and continuous distance.}

Given the apparent difficulty of achieving an efficient constant factor approximation algorithm for the \frechet\ distance between two arbitrary polygonal curves, a natural direction is to develop algorithms for realistic inputs.
Several definitions were suggested to capture and model realistic curves, which include the widely studied $c$-packed curves. This model was introduced by Driemel, Har-Peled and Wenk~\cite{DHW12}, who also presented a $(1+\eps)$-approximation algorithm that runs in $O(cn/\eps + cn \log n)$ time for computing the \frechet\ distance between two $c$-packed curves. Bringmann and Künnemann~\cite{BK17} improved the running time to $O(\frac{cn}{\sqrt{\eps}}\log^2(1/\eps)+cn\log n)$. The discrete \frechet\ distance between two $c$-packed curves was considered in~\cite{HRW24}.
\old{Bringmann~\cite{Bringmann14} showed that, assuming the Strong Exponential Time Hypothesis, there is no $(1+\eps)$-approximation in time $O((cn)^{1-\alpha})$ for a sufficiently small constant $\eps>0$ and any $\alpha>0$. Moreover, in any dimension $d \ge 5$ there is no $(1+\eps)$-approximation in time $O((cn/\sqrt{\eps})^{1-\alpha})$ for any $\alpha>0$.}

Another way to speed-up computation is by using parallel computing with multiple processors. Nowadays, parallel hardware architectures often employ a large number of GPUs for performing massively parallel computations. Such architectures consist of numerous servers, each of which contains several GPUs, which interconnect via high-bandwidth channels. 
The practical analysis of trajectory data on such hardware architectures attracted the attention of researchers and engineers recently. Solutions for \frechet-related problems that significantly benefit from GPU parallelism were obtained in \cite{GV15,HLZL20, MBLG23, MMR19}. While the problem variants are different from those addressed in our current paper, and are not analyzed in our models, they demonstrate the practical importance of massively parallel \frechet\ distance computations. The new results obtained in our paper provide a theoretical foundation, and extend the possibilities of implementing \frechet\ distance algorithms in a parallel-efficient way.

From a theoretical standpoint, very little is known about parallel computation of \frechet\ distance. In their seminal paper, Alt and Godau~\cite{AG95} remark that their divide-and-conquer strategy for closed curves can be adapted to yield a parallel algorithm. However, the details of this reduction appear only in Godau’s Master's thesis~\cite{Godau1991}, which is unfortunately not publicly available.\footnote{We thank Helmut Alt as well as the FU Berlin library staff for their effort in searching for a copy of \cite{Godau1991}.} To the best of our knowledge, no thorough investigation of parallel algorithms for \frechet\ distance has been undertaken to date. In this paper, we initiate a formal study of this important problem. 

We focus on two extensively studied parallel models: the Parallel Random Access Machine (PRAM) model, and the Massively Parallel Computation (MPC) model. For a formal description of these models, see \Cref{sec:prelims}.

\subsection{Results and technical ideas} In this work we focus on parallel algorithms for computing the discrete \frechet\ distance between two polygonal curves, each of length $n$. See Tables~\ref{tbl:results_PRAM},\ref{tbl:results_MPC} for a summary of our results.

\begin{table}[h]
\centering
\bgroup
\def\arraystretch{1.5}
\begin{tabular}{lllllll}
Assumption & Approx.  & Rounds   & Total Work & Remarks &Reference \\ \hline
None & Exact    & $O(n)$  & $ O(n^2)$ & Trivial  & \Cref{obs:trivial_algs}     \\ \hline
None & Exact    & $O(\log{n})$  & $\Tilde{O}(n^6)$ & Trivial  & \Cref{obs:trivial_algs}     \\ \hline
None & Exact    & $O(1)$  & $2^{O(n)}$ & Trivial  & \Cref{obs:trivial_algs}     \\ \hline
None & Exact    & $O(\log{n}/\sigma)$  & $\Tilde{O}(n^{2+\sigma})$ & $\sigma > 0$ (may depend on $n$) & \Cref{thm:exact-dfd}     \\ \hline
$c$-packed & $(1+\eps)$    & $O(\log{n}+1/\sigma^3)$  & $\Tilde{O}(n^{1+\sigma})$ & $\Theta(\frac{\log{c}}{\log{n}}) \leq \sigma < 1$ & \Cref{thm:main-c-packed}
\end{tabular}
\egroup 
\caption{New results in CRCW PRAM}
\label{tbl:results_PRAM}
\end{table}

\begin{table}[h]
\centering
\bgroup
\def\arraystretch{1.5}
\begin{tabular}{llllll}
Assumption & Approx.  & Rounds   & Total Space & Reference \\ \hline
None & Exact    & $O(\log{n}/\rho^2)$  & $O(n^{2+\rho})$ & \Cref{thm:exact-dfd}     \\ \hline
$c$-packed & $(1+\eps)$    & $O(1/\rho^3)$  & $O(n^{1+\rho})$ & \Cref{thm:main-c-packed}
\end{tabular}
\egroup 
\caption{New results in MPC with local space $n^{\rho}$}
\label{tbl:results_MPC}
\end{table}

In all our algorithms, we first design an algorithm for a decision version, where we require the algorithm to determine whether the DFD is larger than $\delta$ (or approximately larger in non-exact variants), and then show how to search for the DFD value given the decision algorithm. To solve the decision versions, we adopt a dynamic programming view of the problem, where we have an $n \times n$ binary matrix $M$, and our goal is to determine whether we have a directed path of 1-entries from $M[1,1]$ to $M[n,n]$, going only rightwards and upwards.

To serve as baseline, we describe three very simple algorithms for this problem (see \Cref{subsec:trivial}).
These algorithms are based on a sequential solution, the best known directed graph reachability solution, and a brute-force solution, respectively. We summarize these results as follows.

\begin{observation}\label{obs:trivial_algs}
    Given two curves of length $n$, their DFD can be computed in the PRAM model in (a) $\Tilde{O}(n)$ rounds and $O(n^2)$ total work, (b) $O(\log{n})$ rounds and $\Tilde{O}(n^6)$ total work, and (c) $O(1)$ rounds and $2^{O(n)}$ total work.
\end{observation}

Our first non-trivial result shows how to compute the exact DFD of two given curves in both the PRAM and MPC models (\Cref{sec:discrete-frechet}). Our algorithm is comprised of $O(\log{n})$ iterations, where we gradually compute the reachability information on this matrix: we introduce a compact data structure that encodes reachability information of a ``strip'' of the matrix (a sequence of consecutive columns), which in particular is sufficient to deduce all reachability relations between every two cells in the leftmost and rightmost columns of the strip respectively. We then show how given the reachability information of two consecutive strips of dimension $n \times 2^{i-1}$, we can compute the reachability information of the merged strip (see \Cref{fig:crossing_paths} for an illustration of this process), which after the last iteration allows us to determine whether $M[1,1]$ can reach $M[n,n]$. 
Finally, we search for the DFD solution using a multi-pivot search on all pairwise distances of the two curves. 

\omritin{TODO: reviewer3 says: Theorem 1: the definition of $\sigma$ is confusing}
\omritin{TODO update the theorem}
\begin{restatable}{theorem}{thmExactDFD}\label{thm:exact-dfd}
    Given two curves of length $n$, their DFD can be computed in the PRAM model in $O(\log{n}/\sigma)$ rounds and $\Tilde{O}(n^{2+\sigma})$ total work for any $\sigma > 0$. In the MPC model, the DFD can be computed in $O(\log{n}/\rho^2)$ rounds and $O(n^{2+\rho})$ total space for any constant $\rho \in (0,1)$.
\end{restatable}
\arnoldin{AF: the LB $\frac{1}{\log{n}} \leq \sigma $ is redundant. Can write $0<\sigma$. I think in all MPC statements that $\rho$ should be implicit, from the model definition. In particular we should stat that in all result, $\rho$ can be replaced by any constant $\rho'<\rho$. How about this formulation:\\
Given two curves of length $n$, their DFD can be computed in the PRAM model in $O(\log{n}/\sigma^2)$ rounds and $O(n^{2+\sigma})$ total work for any $0 \leq \sigma < 1$. In the MPC model, the DFD can be computed in $O(\log{n}/\rho^2)$ rounds and $O(n^{2+\rho})$ total space.
}

\begin{remark}
    As mentioned in the introduction, Alt and Godau~\cite{AG95} present an algorithm for computing the \frechet\ distance between closed curves, which (as they remark in the paper) can be adapted to the PRAM model with the same bounds as in \Cref{thm:exact-dfd}. In fact, they use a similar divide-and-conquer approach. However, their algorithm is slightly more involved because its goal is to compute the distance for closed curves, where the starting point on each curve can be chosen. Moreover, even for the simplified version that we present in this paper for DFD, the adaptation of the divide and conquer algorithm into a parallel algorithm requires some non-trivial steps. Thus, the absence of a publicly available proof of the claim in~\cite{AG95} makes it worthwhile to present a complete parallel algorithm with a formal proof, as we do in this paper.
\end{remark}

Next, we turn our attention to $c$-packed curves, and show that we can find an approximate value of the DFD with greatly reduced total work. 
For this, we first need to simplify the curves, i.e., compute a so-called $\mu$-simplification for each of the two $c$-packed curves. These simplifications allow us to implicitly compute a dynamic programming matrix where we have at most $O(c)$ 1-entries in each row and column. 
The concept of $\mu$-simplifications was introduced in \cite{DHW12} in their algorithm for $c$-packed curves. Such simplifications can be computed using a sequential greedy algorithm. In \Cref{sec:simplification} we define a relaxed version of $\mu$-simplifications, and show how to compute them in a parallel manner. Since simplifications play an important role in many algorithms and applications of curve similarity, we believe that the implementation of such a simplification algorithm in parallel may be of independent interest.

\begin{restatable}{theorem}{thmSimp}\label{thm:MPS-simp}
    In the MPC model, a $(2/\rho+1)$-relaxed $\mu$-simplification of a curve with $n$ vertices can be computed in $O(1/\rho)$ rounds, using $O(n^{\rho})$ space per machine, and with $\Tilde{O}(n)$ total space and total work, for any $\rho < 1$. In the PRAM model, an $O(\log n)$-relaxed $\mu$-simplification of a curve with $n$ vertices can be computed in $O(\log n)$ rounds, $O(n)$ processors, and $O(n\log{n})$ total work.
\end{restatable}

Given the simplification, one can show that if we tile the matrix by submatrices of a given size, only a few of them have 1-entries on the boundary, and hence only few are relevant for determining the DFD. We show how to compute the relevant submatrices, which allows us to reduce the size of the problem. Once the remaining matrix is small enough we either brute-force the solution on the remaining graph, or use the best known algorithm for directed reachability, which solves the decision variant of the problem.

We conclude by finding an approximate solution using the algorithm for the decision problem. We first construct a set $D$ of $\Tilde{O}(n)$ possible candidates for the approximate DFD (which can be considered a structure conceptually resembling but weaker than WSPD trees), and then perform a multi-pivot search for the approximate DFD on this set. 

\begin{restatable}{theorem}{thmCPacked}\label{thm:main-c-packed}
    Given $\eps>0$ and two $c$-packed curves of length $n$, a $(1+\eps)$ approximation of their DFD can be computed in the PRAM model in $O(\log{n}+1/\sigma^3)$ rounds and $O(n^{1+\sigma})$ total work for any $\Theta(\frac{\log{c}}{\log{n}}) \leq \sigma < 1$. In the MPC model, a $(1+\eps)$ approximation of the DFD can be computed in $O(1/\rho^3)$ rounds and $O(n^{1+\rho})$ total space for any constant $\rho \in (0,1)$.
\end{restatable}

\section{Preliminaries}
\label{sec:prelims}
Throughout the paper, we write $\Tilde{O}(f(n))$ to denote $O(f(n)\log^c(f(n)))$ for some sufficiently large constant $c > 0$. A polygonal curve $A:[1,n] \rightarrow \reals^d$ consists of $n-1$ line segments $\overline{a_i a_{i+1}}$ for $i \in \{1,2,\dots,n-1\}$, where $a_i=A(i)$. The points $a_1,\dots,a_n$ are the vertices of $A$. We assume for simplicity of presentation that $d$ is constant, however, note that in all our algorithms it will be easy to see that the dependency on $d$ is linear.

Let $A:[1,n] \rightarrow \reals^d$ and $B:[1,n] \rightarrow \reals^d$ be two polygonal curves. We assume for simplicity that both curves have $n$ vertices, but this assumption can be easily removed.
A \dfn{walk} $\pi$ on $A$ and $B$ is a sequence of pairs of indices $(i_1,j_1), (i_2,j_2),\dots,(i_t,j_t)$, such that:
\begin{enumerate}[(i)]
    \item $(i_1,j_1)=(1,1)$ and $(i_t,j_t)=(n,n)$, and 
    \item for all $1<k\le t$ it holds that $(i_k,j_k)\in \{(i_{k-1}+1,j_{k-1}),(i_{k-1},j_{k-1}+1),(i_{k-1}+1,j_{k-1}+1)\}$.
\end{enumerate}
The \dfn{discrete \frechet\ distance} (DFD) between $A$ and $B$ is defined as $$\dfd(A,B)=\min_{\pi \in \Pi}\max_{(i,j)\in \pi}\|a_i-b_j\|$$ where $\Pi$ is the set of all possible walks on $A$ and $B$. 

Let $M$ be the $n\times n$ matrix such that $M[i,j]=\|a_i-b_j\|$. The $i$-th column of $M$ corresponds to the $i$-th vertex of $A$, and the $j$-th row corresponds to the $j$-th vertex of $B$. A walk $\pi$ on $A$ and $B$ corresponds to a \emph{monotone path} in $M$ from $M[1,1]$ to $M[n,n]$, and thus finding such a monotone path $\pi$ with minimum ``cost'' (i.e. minimum $\max_{(i,j)\in \pi}\|a_i-b_j\|$) is equivalent to computing the discrete \frechet\ distance between $A$ and $B$. This can be done using a simple dynamic programming algorithm in $O(n^2)$ time. Alternatively, this problem can be viewed as finding a directed path with the smallest maximal cell weight from the bottom left vertex to the top right, in the corresponding directed grid graph where each vertex has at most three directed edges: to its top, top-right, and right neighbors (see \Cref{fig:dp}).

\begin{figure}[h]
    \centering
    \includegraphics[page=1, scale=1]{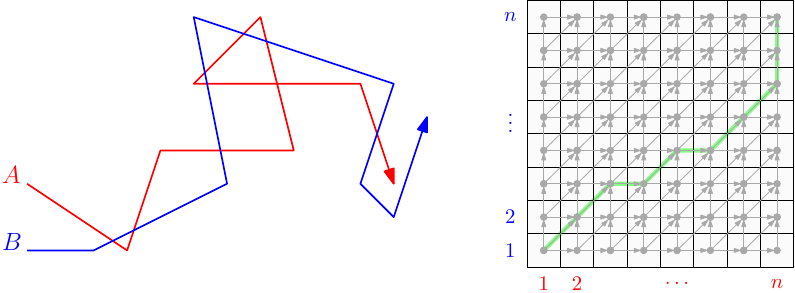}
    \caption{The matrix $M$ of two curves $A$ and $B$, and a monotone path in the corresponding directed grid graph.}
    \label{fig:dp}
\end{figure}

\oldparagraph{The decision version.}
Given a distance threshold $\delta$, let $M_\delta$ be the $n\times n$ binary matrix where $M_\delta[i,j]=1$ if and only if $\|a_i-b_j\|\le \delta$. 
Let $\Pi_\delta$ be the set of all walks $\pi$ on $A$ and $B$ with $\max_{(i,j)\in \pi}\|a_i-b_j\|\le \delta$. In other words, $\Pi_\delta$ is the set of all monotone paths in $M_\delta$ from $M_\delta[1,1]$ to $M_\delta[n,n]$ that use only 1-entries of $M_\delta$. The decision problem can therefore be viewed as a \textbf{reachability problem}, as follows. Consider the directed subgraph of the grid which contains only the vertices that correspond to $1$-entries in $M_\delta$. The discrete \frechet\ distance between $A$ and $B$ is at most $\delta$ if and only if $M_\delta[n,n]$ is reachable from $M_\delta[1,1]$ in this subgraph.

\input{sec_models}

\input{sec_trivialAlgs}

\input{sec_computingDFD}

\section{Simplifying a curve in parallel}\label{sec:simplification}
Curve simplification is a very powerful heuristic that is used in many approximation algorithms in order to reduce running time. Essentially, a simplification of a curve $P$ with $n$ vertices is a curve $P'$ with $k\le n$ vertices such that $\dfd(P,P')$ (or another distance measure) is relatively small, while the curve $P'$ is ``simpler'' in some sense. Many algorithms exist for computing different types of simplifications, each with its own useful properties. In this section we present parallel algorithms for computing a specific type of simplification, which will be useful in the next section for approximating the distance between realistic curves.

In their celebrated work on approximation algorithms for \frechet\ distance between realistic curves, Driemel, Har-Peled and Wenk~\cite{DHW12} introduced the concept of $\mu$-simplifications. The definition that we give here is slightly more general.

\old{
\oldparagraph{$\mu$-simplifications.} A $\mu$-simplification of a polygonal curve $P$ with $n$ vertices is a curve $P^\mu=(p_{i_1},\dots,p_{i_k})$, such that there exists a partition of $P$ into disjoint subcurves $Q_1,\dots,Q_k$ where $p_{i_t}\in Q_t$, and the following properties hold:
\begin{enumerate}[(a)]
    \item Small distance: for every $w \in Q_t$ we have $\|p_{i_t}-w\|\le\mu$ (which immediately implies that $\dfd(P,P^\mu)\le\mu$).
    \item Long edges: for every $t$, $\|p_{i_t}-p_{i_{t+1}}\|>\mu$, i.e., all the edges of $P^\mu$ have length larger than $\mu$.
\end{enumerate}
}

\oldparagraph{A sequential algorithm.} In~\cite{DHW12}, a $\mu$-simplification is defined as the output of the following algorithm. Start with the first vertex $p_1$, add it to $P^\mu$, and set it as the current vertex $p_i$. Then, iterate over the vertices of $P$ that follow $p_i$ to find the first vertex $p_j$ such that $\|p_i-p_j\|>\mu$. Add $p_j$ to $P^\mu$ and set it as the current vertex. Continue the process until reaching the last vertex $p_n$.\footnote{In \cite{DHW12}, the algorithm also adds $p_n$ as the last vertex in $P^\mu$. Here we skip this step.}
Let $P^\mu=(p_{i_1},\dots,p_{i_k})$, and notice that this algorithm also produces a partition of $P$ into disjoint subcurves $Q_1,\dots,Q_k$ such that $p_{i_t}\in Q_t$, for every $w \in Q_t$ we have $\|p_{i_t}-w\|\le\mu$, and all the edges of $P^\mu$ have length larger than $\mu$.

This algorithm is inherently sequential, and to compute a simplification in parallel we would need to merge subcurves of the simplified curves together. For this, we need to allow some relaxation of $\forall_{w \in Q_t}\|p_{i_t}-w\|\le\mu$, as follows.

\oldparagraph{$\alpha$-relaxed $\mu$-simplifications.} For some relaxation parameter $\alpha>1$, an $\alpha$-relaxed $\mu$-simplification $P^\mu_\alpha$ is a polygonal curve $(p_{i_1},\dots,p_{i_k})$ such that $P$ can be partitioned into $k$ disjoint subcurves $Q_1,\dots,Q_k$ where $p_{i_t}\in Q_t$, and the following properties hold.
\begin{enumerate}[(i)]
    \item Small distance: for every $w \in Q_t$ it holds that $\|p_{i_t}-w\|\le\alpha\mu$. \label{property:SmallDistance}
    \item Long edges: for every $t$, $\|p_{i_t}-p_{i_{t+1}}\|>\mu$, i.e., all the edges of $P^\mu_\alpha$ have length larger than $\mu$.\label{property:LongEdges}
\end{enumerate}
We call $Q_1,\dots,Q_k$ an $\alpha$-relaxed $\mu$-partition of $P$.

\subsection{Simplification in the MPC model}
\label{sec:simp_mpc}
\subsubsection*{Warmup: Simplification with $\Tilde{O}(\sqrt{n})$ space per machine and $O(1)$ rounds}
We first describe an algorithm with two rounds in the MPC model, where each machine has $\Tilde{O}(\sqrt{n})$ space. For this, we assume that each machine initially stores a consecutive subcurve of $P$ with $\sqrt{n}$ vertices, i.e., $P$ is divided into $\sqrt{n}$ consecutive subcurves, $P_1,\dots,P_{\sqrt{n}}$, each with $\sqrt{n}$ vertices. This assumption can be guaranteed by performing a standard sorting procedure using $O(1)$ rounds (see \Cref{lem:aggregate}).

\oldparagraph{First round.} For each subcurve, compute a $2\mu$-simplification locally on a single machine, using the algorithm of \cite{DHW12} described above. Denote by $P'_1,\dots,P'_{\sqrt{n}}$ the corresponding $2\mu$-simplifications. Notice that for the concatenation of simplifications $P'=P'_1\circ\dots\circ P'_{\sqrt{n}}$ we get property (\ref{property:SmallDistance}), however, there may be short edges. These short edges may occur only between the last vertex of a simplification $P'_i$ and the first vertex of the consecutive simplification $P'_{i+1}$. Therefore, we proceed to the next round with a sequence of pairs and singletons $\I=(\I_1,\dots,\I_{\sqrt{n}})$, such that for each $i\in[\sqrt{n}]$, if $P'_i$ is a single vertex $v$, then $\I_i$ is a \textbf{singleton} $(v)$, and otherwise, if $P'_i$ has more than one vertex, then $\I_i$ is a \textbf{pair} $(v,u)$ where $v$ is the first vertex and $u$ is the last vertex of $P'_i$.

\oldparagraph{Second round.} We can now process the sequence $\I=(\I_1,\dots,\I_{\sqrt{n}})$ on a single machine. We start with the following \emph{compression procedure}. Consider a maximal consecutive subsequence of singletons in $\I$, and denote it by $S=(\I_j,\dots,\I_k)$. Compute a $2\mu$-simplification $S'$ of $S$ (again using the algorithm of \cite{DHW12}), and replace $P'_j\circ \dots\circ P'_k$ by $S'$. The simplification $S'$ has edges of length larger than $2\mu$, and by the triangle inequality we get that $S'$ is a $2$-relaxed $2\mu$-simplification of $P_j\circ \dots\circ P_k$ (i.e., its edges have length larger than $2\mu$ and its distance to $P_j\circ \dots\circ P_k$ is at most $4\mu$).\omrit{add a little separate observation here?} If $S'$ is a single vertex, replace the subsequence $S$ by a singleton $S'$, and otherwise replace it by a pair consisting of the first and the last vertices of $S'$. Do this for every maximal consecutive subsequence of singletons in $\I$. After this process, the modified sequence $\I$ does not have two consecutive singletons. However, some edges in the modified simplification $P'$ may still be too short. This may happen only between vertices of two consecutive elements in $\I$, so we run the following \emph{merge procedure}. First, if $I_1= u_1$ is a singleton and $I_2=(u_2,u_3)$ is a pair, then delete $u_1$ from $P'$ if $\|u_1-u_2\|\le\mu$. Then, on each pair element $\I_i=(v_1,v_2)$ of $\I$ for $i=1,\dots,|\I|-1$:
\begin{enumerate}
    \item If $\I_{i+1}$ is a pair $(v_3,v_4)$:
    \begin{itemize}
        \item If $\|v_2-v_3\|\le\mu$: delete $v_3$ from $P'$.
    \end{itemize}
    \item Else, if $\I_{i+1}$ is a singleton $v_3$ (then $\I_{i+2}$, if exists, is a pair $(v_4,v_5)$):
    \begin{itemize}
        \item If $\|v_2-v_3\|\le\mu$: delete $v_3$ from $P'$. If $\|v_2-v_4\|\le\mu$: delete $v_4$ from $P'$.
        \item Else ($\|v_2-v_3\|>\mu$), if $\|v_3-v_4\|\le\mu$: delete $v_4$ from $P'$.
    \end{itemize}
\end{enumerate}
Notice that the above procedure can run in parallel for all the pairs, as the computation is done only locally, and only the first element of a pair can be deleted from $P'$. We use this observation only later, when we generalize the algorithm to have a tradeoff between the bound on the space and the number of rounds.

\begin{claim}\label{clm:simplification-2-rounds}
    Let $P'$ be the result of the algorithm after the second round. Then $P'$ is a $5$-relaxed $\mu$-simplification of $P$.
\end{claim}
\begin{proof}
    Recall that every element $\I_i\in \I$ corresponds to a $2$-relaxed $2\mu$-simplification $P'_i$ of $P_i$.\omrit{slight abuse of notation, because we modify $P'$}
    For every $i$ and $v\in P'_i$, denote by $Q(v)$ the subcurve of $P$ that corresponds to $v$ in the $2$-relaxed $2\mu$-partition of $P_i$, so $\|v-w\|\le 4\mu$ for every $w\in Q(v)$.
    Therefore, if $\|v-u\|\le\mu$ for some vertices $u,v\in P'$, then by the triangle inequality all the vertices in $Q(u)$ have distance at most $5\mu$ to $v$. Since we only delete a vertex $u$ if $\|v-u\|\le\mu$ for the vertex $v$ that precedes it, we get that property (\ref{property:SmallDistance}) holds for the modified $P'$. It remains to show that property (\ref{property:LongEdges}) holds.
    By definition, for every $i$, all the edges of $P'_i$ have length larger than $2\mu$. First, we note that if the first element $I_1 = u_1$ is a singleton and is not deleted, then for $I_2 = (u_2,u_3)$ we have $\|u_1-u_2\|>\mu$, so the first edge satisfies the long-edge
property. Next, we show that after running the merge procedure there are no short edges between two consecutive pair elements.  Consider a pair element $\I_i=(v_1,v_2)$. There are two cases:
    \begin{enumerate}
        \item $\I_{i+1}$ is a pair $(v_3,v_4)$, and $\|v_2-v_3\|\le\mu$. Let $w$ be the vertex that comes after $v_3$ in $P'_{i+1}$. Because $P'_{i+1}$ is a $2$-relaxed $2\mu$-simplification, we have $\|v_3-w\|> 2\mu$, and by the triangle inequality, $\|v_2-w\|>\mu$.
        \item $\I_{i+1}$ is a singleton $v_3$, and $\I_{i+2}$ is a pair $(v_4,v_5)$. Let $w$ be the vertex that comes after $v_4$. If $\|v_2-v_3\|\le\mu$ and $\|v_2-v_4\|\le\mu$ then both $v_3,v_4$ are deleted, and as in the previous case, $\|v_2-w\|>\mu$ because $\|v_4-w\|> 2\mu$. If $\|v_2-v_3\|>\mu$ and $\|v_3-v_4\|\le\mu$, then by similar arguments $\|v_3-w\|>\mu$.
    \end{enumerate}
    Note that if $\I_{i+1}$ is a singleton $v_3$ and it is the last element in $\I$, then $v_3$ is deleted if $\|v_2-v_3\|\le\mu$, and $v_2$ will be the last vertex of $P'$.
\end{proof}

\oldparagraph{Generalization to MPC with $n^\rho$ space and $O(1/\rho)$ rounds, and generalization to PRAM.} The approach above can be naturally generalized for machines with local space $O(n^\rho)$ for any parameter $\rho \in (0,1)$. We briefly sketch the altered approach, and defer the exact details to \Cref{sec:app_simple_mpc}. Similarly, the same algorithm translates in an essentially identical manner to the PRAM model, and we defer the details to \Cref{sec:simp_PRAM_appendix}.

On a high level, the algorithm works as follows. We define a hierarchical partitioning of the curve $P$, where in level $j \in \{0,\dots,1/\rho\}$, the curve $P$ is divided into $n^{1-j\rho}$ consecutive subcurves, $P^j_1,\dots,P^j_{n^{1-j\rho}}$. The algorithm performs ``repeated simplification'' for $1/\rho$ rounds, where in round $j$, we have for each $P^j_i$ a designated machine $i$. This machine receives compressed information $\I(P^j_i)$ about $P^j_i$, where $\I(P^j_i)$ is a sequence containing $O(n^\rho)$ elements of singleton vertices of $P^j_i$ and vertex pairs of $P^j_i$; here a pair signifies a long edge. Intuitively, $\I(P^j_i)$ is a partial $(2j)$-relaxed $\mu$-simplification of $P^j_i$, where a vertex of $P^j_i$ not appearing in $\I(P^j_i)$ was either simplified away by some element in $\I(P^j_i)$, or in some sense we already found for the vertex its long edges, and it was removed from the recursion. The machine replaces each consecutive subsequence of singletons of $\I(P^j_i)$ with a $2\mu$-simplification of it, and (up to some technical details) sends the two extremal elements of the remaining elements to the next level. After the final level, we collect the information gathered in all rounds together, in a manner that guarantees that no consecutive singletons remain. Finally, we use the same merge procedure as described above, and obtain a $(2/\rho+1)$-relaxed $\mu$-simplification. Exactly the same approach works for the PRAM setting, where we obtain an $O(\log{n})$-relaxed $\mu$-approximate simplification in $O(\log{n})$ rounds.

\thmSimp*

\old{
\begin{remark}
    We can get a 3-relaxed, see simplification.tex\omrit{remove this for SODA probably, but keep for later}
\end{remark}
}

\section{Realistic curves}\label{sec:c-packed}
A curve $P:[1,n]\rightarrow \reals^d$ is \emph{$c$-packed} if for any ball $B$ with radius $r$, the total length of $P\cap B$ is at most $c\cdot r$.
As mentioned in the previous section, the concept of $c$-packed was first introduced by Driemel, Har-Peled and Wenk \cite{DHW12}, who also gave a near linear time $(1+\eps)$ approximation for computing the (continuous) \frechet\ distance between two $c$-packed curves. The high level idea of their algorithm is based on computing $\mu$-simplifications of the original curves, on which the computation of \frechet\ distance can be done in near linear time. We outline their approach in more detail below.

Recall that a $\mu$-simplification $P^\mu$ of a curve $P$ with $n$ vertices can be computed in linear time, and has the following properties (see \cite{DHW12}):
(a) $\dfd(P,P^\mu)\le\mu$, and (b) all the edges of $P^\mu$ have length larger than $\mu$. 
As mentioned in the introduction, the decision version of the discrete \frechet\ distance can be viewed as a reachability problem, and thus the running time for traversing the matrix $M_\delta$ is linear in its ``complexity'', i.e., the number of 1-entries in $M_\delta$. Driemel, Har-Peled and Wenk \cite{DHW12} show that for two $c$-packed curves whose edges have length at least $\mu$, the complexity of the $\delta$-reachable free space (a similar concept defined for the continuous distance) is $O(\frac{c\delta}{\mu}n)$. In addition, they show that if $P$ is $c$-packed, then $P^\mu$ is $6c$-packed. This leads to the following approximate decision algorithm. Given two $c$-packed curves $P,Q$ and values $\eps,\delta$, construct $P^{\frac{\eps\delta}{4}}$ and $Q^{\frac{\eps\delta}{4}}$ in linear time. Then, run the exact decision algorithm to check whether $\dfd(P^{\frac{\eps\delta}{4}},Q^{\frac{\eps\delta}{4}})\le\delta'=(1+\frac{\eps}{2})\delta$. Since the complexity of the $\delta$-reachable free space is $O(\frac{c\delta'}{\mu}n)$, this can be done in $O(cn/\eps)$ time.
By property (a) and the triangle inequality, if the answer is yes then $\dfd(P,Q)\le(1+\eps)\delta$, and otherwise $\dfd(P,Q)>\delta$.

Below, we first show that similar properties hold for relaxed simplifications, and then present parallel algorithms with near linear\omrit{explicit?} total work for computing a $(1+\eps)$-approximation for the discrete \frechet\ distance between two $c$-packed curves. \omritin{say something about the challenges: the reachability algorithm is sequential, we can evaluate the entire matrix, optimization}

The following lemma helps us bound the complexity of the binary matrix on which we solve our decision problem. Due to space considerations, we defer its proof to \Cref{sec:packedness_missing_proofs}. 
\begin{lemma} \label{lem:matrix-complexity}
    Let $P$ and $Q$ be two polygonal curves, both $c$-packed.
    Given a distance threshold $\delta$ and a parameter $\eps>0$, $\alpha\ge 1$, set $\mu=\frac{\eps\delta}{4\alpha}$. Let $M=M_{(1+\frac{\eps}{2})\delta}$ be the binary matrix that corresponds to $P'$ and $Q'$, where $P',Q'$ are $\alpha$-relaxed $\mu$-simplifications of $P,Q$, respectively. Then the number of 1-entries in each row and each column of $M$ is $O(\frac{\alpha^2c}{\eps})$.
\end{lemma}

We also need the lemma below, which allows us to ignore large portions of $M$. Consider a submatrix $M'$ of the matrix $M$. The \textbf{boundary cells} of $M'$ are all the cells in its first and last row, and the cells in its first and last column. We say that a submatrix $M'$ is \textbf{relevant} if there is at least one 1-entry among its boundary cells.

\begin{lemma}\label{lem:submatrices-complexity}
    Let $M$ be an $n\times n$ binary matrix such that in each row and each column there are at most $\lambda$ 1-entries. Given $k\le n$, for every $i,j\in \{1,\dots,k\}$, let $M^{i,j}$ be the submatrix of $M$ that contains all the cells $M[a,b]$ for $a\in [(i-1)\frac{n}{k}+1,i\frac{n}{k}]$ and $b\in [(j-1)\frac{n}{k}+1,j\frac{n}{k}]$. Then at most $4\lambda k$ of these submatrices are relevant (i.e. contain a 1-entry on their boundary cells).
\end{lemma}
\begin{proof}
    The total number of 1-entries in boundary cells is at most $4\lambda k$ (because all the boundary cells are contained in $2k$ rows and $2k$ columns, each has at most $\lambda$ 1-entries). So there are at most $4\lambda k$ matrices that contain a 1-entry on their boundary.
\end{proof}

\subsection{Finding the relevant submatrices}
Next, we show that when setting $k=n^{1-\rho}$ \orr{but in the next paragraph we set $k = n^{\rho}$ instead}\omrit{Right, the idea is to set $k=n^\rho$ in each round and eventually get the matrices that correspond to $k=n^{1-\rho}$. Maybe we should emphasize this...}, the set of relevant submatrices can be found in the MPC model with $O(n^\rho)$ space per machine, in $O(1/\rho)$ rounds and using $O(n^{1+\rho})$ total space. The idea is to apply \Cref{lem:submatrices-complexity} in each round with $k=n^\rho$ on all the submatrices found in the previous round, until we get the submatrices that correspond to setting $k=n^{1-\rho}$ in the lemma.

\oldparagraph{First round.} Consider all the submatrices $\{M^{i,j}\}_{i,j \in [k]}$ described in \Cref{lem:submatrices-complexity} for the case of $k=n^\rho$. The boundary cells of these matrices are contained in $2n^\rho$ rows and $2n^\rho$ columns of the original matrix $M$, so in total $O(n^{1+\rho})$ cells. 
We check all these cells in one round using $n$ machines, each with $O(n^\rho)$ memory. 
By \Cref{lem:submatrices-complexity}, at most $4\lambda n^\rho$ of them are relevant. Each submatrix has $O(n/k)=O(n^{1-\rho})$ rows and columns.

\oldparagraph{Second round.} Apply \Cref{lem:submatrices-complexity} on each of the submatrices that were found in the first round, again with $k=n^\rho$. In each submatrix we get $4 n^\rho$ relevant rows and columns, and in total $O(n)$ boundary cells per submatrix (each row and column now has $O(n^{1-\rho})$ cells). In total over all the relevant submatrices from the first round, there are $O(\lambda n^{1+\rho})$ cells to check. By \Cref{lem:submatrices-complexity}, at most $4\lambda n^{2\rho}$ of the submatrices are relevant (dividing each submatrix again is as if we applied the lemma with $k^2=n^{2\rho}$). Each submatrix has $O(n/k^2)=O(n^{1-{2\rho}})$ rows and columns. 

\oldparagraph{Round $i$.} At the beginning of the $i$-th round we get at most $4\lambda n^{(i-1)\rho}$ submatrices, each with $O(n^{1-{(i-1)\rho}})$ rows and columns. By applying \Cref{lem:submatrices-complexity} with $k=n^\rho$ on each of them, we get $O(n^{1-{(i-1)\rho}}\cdot n^\rho)=O(n^{1-{(i-2)\rho}})$ cells to check in each submatrix, so in total $O(n^{1-{(i-2)\rho}}\cdot \lambda n^{(i-1)\rho})=O(\lambda n^{1+\rho})$ cells to check. We are then left with at most $4\lambda n^{i\rho}$ relevant submatrices, each has $O(n^{1-{(i-1)\rho}}/n^\rho)=O(n^{1-{i\rho}})$ rows and columns.

After $\frac{1-\rho}{\rho}=O(1/\rho)$ rounds, we are left with at most $4\lambda n^{1-\rho}$ relevant submatrices (that contain a 1-entry on their boundary), each has $O(n^\rho)$ rows and columns. Notice that by the recursiveness of the algorithm, these are a subset of the relevant matrices that we would get from applying \Cref{lem:submatrices-complexity} on $M$ with $k=n^{1-\rho}$, where any removed relevant submatrix is necessarily contained in some non-relevant submatrix (since we only recurse when a submatrix has 1-entries on its boundary). We conclude that the set of at most $4\lambda n^{1-\rho}$ relevant matrices can be found in $O(1/\rho)$ rounds, using $O(n^\rho)$ space per machine, and $O(n^{1+\rho})$ total space.


We notice that this procedure translates directly to the PRAM setting as well, where the role of each MPC machine in each round can be replaced by $O(n^{\rho})$ processors, acting in $O(1)$ rounds. This is since in each round a machine just checks the boundary of some sub-matrices for 1-entries, where each boundary cell can be checked independently by some processor. 

\begin{theorem}\label{thm:relevant-submatrices}
    Let $M$ be an $n\times n$ binary matrix such that in each row and each column there are at most $\lambda$ 1-entries, for some constant $\lambda\ge 1$. 
    For every $i,j\in [1,\dots,n^{1-\rho}]$, let $M^{i,j}$ be the submatrix of $M$ that contains all the cells $M[a,b]$ for $a\in [(i-1)n^\rho+1,i\cdot n^\rho]$ and $b\in [(j-1)n^\rho+1,j\cdot n^\rho]$. In the MPC model, for any constant $\rho < 1$, we can compute in $O(1/\rho)$ rounds, using $O(\lambda n^\rho)$ space per machine and $O(\lambda n^{1+\rho})$ total space a subset of $\{M^{i,j}\}_{i,j \in [n^{1-\rho}]}$ of size at most $4\lambda n^{1-\rho}$, where any $M^{i,j}$ that is not in the output is either not relevant, or is contained in a non-relevant submatrix in $M$.
    
    Similarly, in the PRAM setting we can compute such a set of submatrices in $O(1/\rho)$ rounds, using $O(\lambda n^{1+\rho})$ processors, for any constant $\rho < 1$. \orr{changed here}
\end{theorem}

\subsection{The decision problem}
We now have all the building blocks that we need in order to present an approximate decision algorithm. We follow the same approach as in~\cite{DHW12}.

Given two polygonal curves $P$ and $Q$, each of length $n$, and parameters $\eps,\delta>0$ and $\alpha > 1$, consider the following (sequential) algorithm:
\begin{enumerate}
    \item Let $\mu=\frac{\eps\delta}{4\alpha}$. 
    \item Compute $\alpha$-relaxed $\mu$-simplifications for $P$ and $Q$, and denote them by $P^\mu$ and $Q^\mu$, respectively.
    \item If $\dfd(P^\mu,Q^\mu)\le (1+\frac{\eps}{2})\delta$, return YES, and otherwise return NO.
\end{enumerate}
\begin{claim}
    If the algorithm above returns YES, then $\dfd(P,Q)\le (1+\eps)\delta$, and if the algorithm returns NO, then $\dfd(P,Q)> \delta$.
\end{claim}
\begin{proof}
Notice that $2\alpha\mu=\frac{\eps\delta}{2}$. By the triangle inequality, if $\dfd(P^\mu,Q^\mu)\le (1+\frac{\eps}{2})\delta$ 
then 
$$\dfd(P,Q)\le \dfd(P,P^\mu)+\dfd(P^\mu,Q^\mu)+\dfd(Q,Q^\mu)\le \alpha\mu+(1+\frac{\eps}{2})\delta+\alpha\mu=(1+\eps)\delta.$$
And, if $\dfd(P^\mu,Q^\mu) > (1+\frac{\eps}{2})\delta$ then
$$\dfd(P,Q)> \dfd(P^\mu,Q^\mu)-\dfd(P,P^\mu)-\dfd(Q,Q^\mu)> (1+\frac{\eps}{2})\delta-\alpha\mu-\alpha\mu=\delta.$$
\end{proof}

Note that if $\dfd(P,Q)\in(\delta,(1+\eps)\delta]$ then the algorithm may return either YES or NO. In ~\cite{DHW12}, the above \textbf{approximate decider} is then turned into a so-called \textbf{``fuzzy'' decider}, by slightly adjusting $\eps$ and $\delta$ and calling the above algorithm (at most twice) with different parameters. This fuzzy decider can then be used in a binary search algorithm. Below, we show how to implement the approximate decider in parallel.
Set $\alpha=2/\rho+1$, and $\mu=\frac{\eps\delta}{4\alpha}$.

\oldparagraph{Step 1: Compute simplifications.} First, compute $\alpha$-relaxed $\mu$-simplifications $P'$ and $Q'$ of $P$ and $Q$, respectively, using the algorithm from \Cref{thm:MPS-simp}.
Let $M=M_{(1+\frac{\eps}{2})\delta}$ be the $n\times n$ binary matrix that corresponds to $P'$ and $Q'$. Then by \Cref{lem:matrix-complexity}, the number of 1-entries in each row and column of $M$ is $\lambda=O(\frac{\alpha^2c}{\eps})$.

\oldparagraph{Step 2: Find the relevant submatrices.} Then, run the algorithm from \Cref{thm:relevant-submatrices}, which computes a set $S$ of at most $4\lambda n^{1-\rho}$ relevant submatrices.

\oldparagraph{Step 3: Compute the reachability information in relevant submatrices.} Next, compute the reachability information for each relevant submatrix from $S$ on a single machine:
Each relevant submatrix $M'\in S$ has $n^\rho$ rows and columns, and at most $4\lambda$ 1-entries on its boundary. The machine only gets the vertices that correspond to these rows and columns, so the space remains $n^\rho$. It can then construct $M'$ explicitly, and for every 1-entry on the boundary of $M'$, compute and store the set of 1-entries on the boundary of $M'$ that are reachable from it. Since there are at most $4\lambda$ 1-entries on the boundary of $M'$, the total amount of information stored on a single machine after this process is $O(\lambda^2)$.

\oldparagraph{Step 4: Aggregate the reachability information.} We now aggregate this information from all the relevant submatrices in $S$ as follows. Consider a partition of $M$ into submatrices of size $n^{\rho+\rho/2}\times n^{\rho+\rho/2}$. By \Cref{thm:relevant-submatrices}, the number of relevant submatrices from this larger partition (that may contain some relevant submatrices from $S$) is at most $4\lambda n^{1-3/2\rho}$. We compute those relevant submatrices using the algorithm from \Cref{thm:relevant-submatrices}. Notice that this algorithm produces a hierarchy of relevant submatrices of different sizes, where in each level of the hierarchy the number of relevant submatrices is proportional to the number of rows and columns that define the partition.\omrit{very hand wavy, sorry, improve for next version} We can thus use this hierarchy of relevant submatrices to recursively aggregate the reachability information from relevant submatrices in $S$.

In the first level of the hierarchy, each submatrix $M_i$ has size $n^{\rho+\rho/2}\times n^{\rho+\rho/2}$, and thus contains a subset $S_i\subseteq S$ of at most $n^\rho$ relevant submatrices. We then send the reachability information of the matrices from $S_i$ to a single machine, which can compute the reachability information for the boundary of the larger submatrix $M_i$. Recall that there are again only $4\lambda$ 1-entries on the boundary of $M_i$, so we are again left with a total of $O(\lambda^2)$ space for the reachability information in each submatrix $M_i$.
We continue this process recursively, using the matrices of size $n^{2\rho}\times n^{2\rho}$, so each of them contains at most $n^\rho$ submatrices of size $n^{\rho+\rho/2}\times n^{\rho+\rho/2}$ from the first round.\omrit{TODO for the next version, elaborate} In the last round, we will have all the reachability information in a single machine, which can decide if there exists a path from $M[1,1]$ to $M[n,n]$. The total number of rounds will then be $O(1/\rho)$.

\oldparagraph{}After turning the approximate decider into a fuzzy decider (as in~\cite{DHW12}), we obtain the following theorem.
\begin{theorem}
\label{thm:realistic_mpc_descision}
    Let $\rho < 1,\eps > 0$ be constants, and assume that $\frac{\log{c}}{\log{n}}$ is sufficiently small compared to $\rho,\eps$. Given two $c$-packed curves $P$ and $Q$ and a threshold $\delta \geq 0$, there is an algorithm in the MPC model with space bound $O(n^\rho)$, total space $O(n^{1+\rho})$, and $O(1/\rho^2)$ rounds, that returns one of the following outputs: (i) a $(1+\eps)$-approximation of $\dfd(P,Q)$, (ii) $\dfd(P,Q)<\delta$, or (iii) $\dfd(P,Q)>\delta$. The answer returned is correct.
\end{theorem}

In the PRAM model, we adopt a similar approach, as described in \Cref{sec:decision_PRAM}. 

\begin{theorem}
\label{thm:realistic_pram_descision}
    Given two $c$-packed curves $P$ and $Q$, constant parameters $\eps>0$, $\delta>0$, and a value $\Theta(\frac{\log{c}}{\log{n}}) \leq \sigma<1$, there is an algorithm in the PRAM model with $O(\log{n} + 1/\sigma^2)$ rounds and $\Tilde{O}(n^{1+\sigma})$ total work, that returns one of the following outputs: (i) a $(1+\eps)$-approximation of $\dfd(P,Q)$, (ii) $\dfd(P,Q)<\delta$, or (iii) $\dfd(P,Q)>\delta$. The answer returned is correct.
\end{theorem}

\subsection{Optimization}
Next, we turn our attention to finding a $(1+\eps)$-approximation of the DFD. Unlike our algorithm for general curves, computing all $\Theta(n^2)$ distances and performing a multi-pivot search is too costly with regard to the total space and total work we aim towards. Instead, we first compute a set $D$ of size $|D| = \Tilde{O}(nd)$, that contains a good approximation of all possible distance values between pairs of points on the curves. More accurately, for every pair of points $x,y$ on the two curves, there exists a value $r \in D$ for which $r \leq \|x-y\|_2 \leq (1+\eps)r$. Given such a set $D$, we can perform a multi-pivot search and find a $(1+\eps)$-approximation of the DFD, using the algorithm for the decision version. In the spirit of \cite{FF23} (see Lemma 8.1), we first show how to construct such a set $D$ for points on the line, i.e. $d = 1$ (with additional relaxed guarantees), and then reduce the higher dimension case to the single dimension case.

We conclude with the following theorem, and the full proof is given in \Cref{sec:optimization_missing_proofs}.

\thmCPacked*

\old{
\begin{theorem}
\label{thm:main_finding_cpack_MPC}
    Let $\rho < 1,\eps > 0$, be constants, and assume that $\frac{\log{c}}{\log{n}}$ is smaller than some constant. Given two $c$-packed curves of length $n$, a $(1+\eps)$-approximation of their DFD can be found in the MPC model in $O(1/\rho^3)$ rounds, using $n^{\rho}$ space per machine, and using $\Tilde{O}(n^{1+\rho})$ total space. 
\end{theorem}

\begin{theorem}
\label{thm:main_finding_cpack_PRAM}
    For any constant $\eps > 0$ and value $\Theta(\frac{\log{c}}{\log{n}}) \leq \sigma < 1$, given two $c$-packed curves of length $n$, a $(1+\eps)$-approximation of their DFD can be found in the PRAM model in $O(\log{n}+1/\sigma^3)$ rounds, and $\Tilde{O}(n^{1+\sigma})$ total work.
    
\end{theorem}
}

\section*{Conclusion}
In this paper, we provide the first parallel algorithms for the discrete \frechet\ distance in both the MPC and PRAM models. In addition, we present more efficient parallel algorithms for $c$-packed curves, in both models. Our algorithms for $c$-packed curves also require the implementation of a simplification algorithm in parallel, which may be of independent interest. Our algorithms are close to optimal in total work, given the SETH-based lower bounds in $\cite{Bringmann14}$. 

As mentioned in the introduction, Alt and Godau~\cite{AG95} presented an algorithm for computing the continuous \frechet\ distance between closed curves, which uses a similar divide-and-conquer approach for the decision procedure: they split the free space diagram into smaller diagrams, on which they recursively compute the reachability information for the entire boundary. They do this by partitioning the boundary into intervals, and computing the lowest and highest point on the opposite side of the boundary which is reachable from each interval.
We remark that a slightly simpler approach can be taken in the case of deciding the continuous \frechet\ distance, by splitting the diagram into strips, as we do for the discrete case. Then, we only need to recursively compute the reachability information for the right and left boundary of each strip, by dividing the boundary into intervals as in~\cite{AG95}. The parallel algorithms in the PRAM and MPC model are then similar to the algorithms that we describe for deciding the discrete version.

Clearly, as this is the first paper on this problem, many questions remain open. 
We state two of them here, which we believe are the natural extension of our work. 
The first open question is computing the (discrete) \frechet\ distance between two curves such that only one of them is $c$-packed. This problem was considered in the sequential model by van der Hoog, Rotenberg and Wong~\cite{HRW24} for the discrete distance, and by Gudmundsson, Mai and Wong~\cite{GMW24} for the continuous distance. In both cases there are near-linear time algorithms, and it will be interesting to see if we can achieve efficient parallel algorithms in this case.

The second open question is whether the discrete \frechet\ distance for general curves can be computed in $O(1)$ MPC rounds. One potential avenue is to adapt the approach of \cite{CZ24}, which shows constant round algorithms for dynamic programming problems by reducing them to graph problems in planar graphs with known node embeddings. However, the techniques of \cite{CZ24} do not seem to readily extend to the type of graph problem our dynamic program corresponds to, namely to either directed reachability, or an exact undirected shortest path problem.

\section*{Acknowledgments}
Omrit Filtser is partially supported by the Israel Science Foundation
    (grant No.~2135/24). Arnold Filtser is partially supported by the Israel Science Foundation
    (grant No.~1042/22). Orr Fischer is partially supported by the Israel Science Foundation (grant No.~1042/22 and~800/22).
    
\oldparagraph{AI Disclosure:} All ideas contained in the paper are of the authors. AI was only used for proofreading and minor editing.

\bibliographystyle{alpha}
\bibliography{refs}

\appendix

\input{sec_computingDFD-apx}

\section{Simplifying a curve in parallel - missing proofs}
\label{sec:app_simple}
\subsection{Simplification in the MPC model}
\label{sec:app_simple_mpc}
In this section, we generalize our approach from \Cref{sec:simp_mpc} such that the space bound for each machine is $O(n^\rho)$, for some parameter $\rho < 1$. We assume w.l.o.g. that $1/\rho$ is an integer, reducing $\rho$ if necessary. The following definition will be useful to us in the description of the algorithm.

\begin{definition}\label{def:pairenc}
    Given a sequence of points $C = (v'_1,\dots,v'_k)$, we define the pair encoding of $C$ as the sequence of pairs $\left((v'_1,v'_2),(v'_2,v'_3),\dots,(v'_{k-1},v'_k)\right)$ if $k \geq 2$, and as $v'_1$ if $k = 1$.    
\end{definition}

Next, we describe the algorithm. Consider the following hierarchical partition of the vertices of $P$:

\begin{itemize}
    \item \emph{Level $0$:} The $n$ vertices of $P$.
    \item \emph{Level $1$:} Partition $P$ into $n^{1-\rho}$ consecutive subcurves $P^1_1,\dots,P^1_{n^{1-\rho}}$, each containing $n^\rho$ vertices.
    \item \emph{Level $2$:} Partition $P$ into $n^{1-2\rho}$ consecutive subcurves $P^2_1,\dots,P^2_{n^{1-2\rho}}$, where
    $P^2_i = P^1_{(i-1)n^\rho+1}\circ\cdots\circ P^1_{i\cdot n^\rho}$; that is, each $P^2_i$ is the concatenation of $n^\rho$ consecutive level-$1$ subcurves.
    \item \emph{Level $j \geq 1$:} More generally, partition $P$ into $n^{1-j\rho}$ consecutive subcurves $P^j_1,\dots,P^j_{n^{1-j\rho}}$, each formed by concatenating $n^\rho$ consecutive subcurves from level $j-1$. Hence there are $1/\rho$ levels in total, and level $1/\rho$ consists of the full curve $P$.
\end{itemize}

We initially have $1/\rho$ ``repeated simplification'' rounds (one for each level, starting from level one), followed by a single merge round. In round $j \geq 1$ of the ``repeated simplification'', we only use $n^{1-j\rho}$ machines, where the $i$-th machine is assigned to handle the subcurve that corresponds to $P^j_i$. Each machine $i$ receives a set $\I(P^j_i)=(\I_1,\dots,\I_k)$ of $O(n^\rho)$ elements. For $j = 1$, we set $\I(P_i^1) = P_i^1$, where each vertex is considered by the machine as a ``singleton''. In subsequent rounds the sequence $\I(P_i^j)$ may contain as elements either singleton vertices, or vertex pairs, where all vertices belong to $P^j_i$ and appear in the sequence in an order consistent to that of $P^j_i$. The machine runs a compression procedure, where it computes a $2\mu$-simplification for each maximal consecutive subsequence of singletons, and replaces it with the pair encoding (\Cref{def:pairenc}) of the simplification. We denote the modified sequence $\I'(P^j_i)$, and note that it does not have two consecutive singletons. 

Next, we define $\I(P^{j+1}_{i})$ from $\I'(P^j_{(i-1)n^\rho+1}),\dots,\I'(P^j_{i \cdot n^\rho})$ as follows. For every $\ell \in [(i-1)n^\rho+1,i \cdot n^\rho]$, if $\I'(P^j_{\ell})$ has a single element $\I'(P^j_{\ell}) = \{\I_1\}$, then we add $\I_1$ to $\I(P^{j+1}_{i})$. 

Else, if $\I'(P^j_\ell) = (\I_1,\dots,\I_k)$ has at least two elements, then we add the leftmost and rightmost elements $\I_1,\I_k$ to $\I(P^{j+1}_{i})$, and in the case that $\I_1,\I_k$ are both singletons, we additionally add $\I_2$ to $\I(P^{j+1}_{i})$. We note that $\I_2$ is necessarily a pair since $\I'(P^j_\ell)$ does not contain consecutive singletons (as noted earlier).\footnote{We took here an arbitrary pair, because pairs are ignored and are only used as ``breakpoints'' between sequences of singletons.} As the final step in round $j$, the machine locally stores all elements in $\I'(P^j_\ell)$ not sent to $\I(P^{j+1}_i)$ for the final merge round. We say that these elements are added to $P'$ in this step. 

We notice that at most three elements are added to $\I(P^{j+1}_{i})$ from each machine responsible for $P^j_{\ell}$ for $\ell \in [(i-1)n^\rho+1,i \cdot n^\rho]$. In a single round, each machine responsible for $P^j_{\ell}$ sends the relevant elements to the machine responsible for $P^{j+1}_{i}$. Therefore, the machines of the next level receive $O(n^\rho)$ elements in total. We repeat until we reach the final level, i.e. $1/\rho$ rounds. In the final round $1/\rho$, we have a single machine holding all remaining elements. The single active machine computes a $2\mu$-simplification for each maximal consecutive singleton sequence and stores all elements in $\I'(P^{1/\rho}_1)$ to be in $P'$. 

In the merge round, we define the sequence $\I=(\I_1,\dots,\I_k)$ of $O(n)$ elements (pairs and singletons), where $\I$ consists of the elements stored for $P'$ in the simplification rounds. We note that there are no two consecutive singletons in $\I$. Using the machines, we run the single-round merge procedure as described in \Cref{sec:simp_mpc}. Namely, we do the following. If the first element of $\I$ is a singleton $\I_1 = v_1$, and $\I_2 = (v_2,v_3)$ where $\|v_1-v_2\| \leq \mu$, then we remove $v_1$ from $P'$. We then perform the following on each pair element $\I_i=(v_1,v_2)$ of $\I$ for $i=1,\dots,|\I|-1$:
\begin{enumerate}
    \item If $\I_{i+1}$ is a pair $(v_3,v_4)$:
    \begin{itemize}
        \item If $\|v_2-v_3\|\le\mu$: delete $v_3$ from the pair.
    \end{itemize}
    \item Else, if $\I_{i+1}$ is a singleton $v_3$ (then $\I_{i+2}$, if exists, is a pair $(v_4,v_5)$):
    \begin{itemize}
        \item If $\|v_2-v_3\|\le\mu$: delete $v_3$. If $\|v_2-v_4\|\le\mu$: delete $v_4$ from the pair.
        \item Else ($\|v_2-v_3\|>\mu$), if $\|v_3-v_4\|\le\mu$: delete $v_4$ from the pair.
    \end{itemize}
\end{enumerate}

We note that this step can be easily implemented using the sorting procedure of MPC in $O(1)$ rounds. This concludes the description of the algorithm. Next, we show that the resulting sub-sequence $P'$ after the merge procedure is a $(2/\rho+1)$-relaxed $\mu$-simplification of $P$.

Recall that in an $\alpha$-simplification of a curve $R$, we get a subset of vertices $r_1,\dots,r_k$ such that $R$ can be partitioned into disjoint subcurves $R_1,\dots,R_k$, with $r_i \in R_i$, satisfying the small distance and long edge properties. We say that $v$ represents $u$ in a simplification if $v$ remains at the end of the simplification and the subcurve $R_i$ corresponding to $v$ contains $u$.

\begin{claim}
\label{clm:app_simp_mpc_many_prop}
    The following properties hold at the start of the merge process.
    \begin{enumerate}[(1)]
        \item \label{itm:app_simp_long_edge} If a pair $\I = (v_1,v_2)$ is added to $P'$ at any round, then $\|v_1-v_2\| > 2\mu$.
        \item \label{itm:app_simp_short_edge} Let $p_1,\dots,p_k \in P$ be the vertices participating in elements of $P'$, ordered by their order in $P$. There exists a partition of $P$ into disjoint subcurves $Q_1,\dots,Q_k$ such that for every $w \in Q_i$, we have $\|p_i-w\| \leq (2/\rho) \mu$.
        
    \end{enumerate}
\end{claim}
\begin{proof}
    We show each property holds.
    \begin{enumerate}[(1)]
        \item A pair $(v_1,v_2)$ is only created after $v_1,v_2$ participated in a $2\mu$-simplification process together, and both were not deleted in the simplification. Hence $\|v_1-v_2\| > 2\mu$.
        \item We define for purposes of analysis the following virtual forest. To avoid confusion, we use the term nodes for nodes of the forest, and vertices for parts of the curve $P$.
        
        The forest has $(1/\rho)+1$ levels, numbered $0,\dots,1/\rho$. In level 0, we have each of the $n$ vertices of $P$ as a node, and in level $j \geq 1$ add a node for each of the vertices participating in an element of $\I'(P^j_i)$ that replaced a sequence of singletons following a $2\mu$-simplification. Edges in the forest may only connect consecutive levels, and are described next. 
        
         We connect a node corresponding to vertex $u$ in level $j$ to a node corresponding to vertex $v$ in level $j+1$ if $u,v \in \I(P^{j+1}_i)$ for some $i$, they participate together in a $2\mu$-simplification in the machine associated with $P^{j+1}_i$, and $v$ represents $u$ in this simplification.

         In the algorithm, singletons in $\I(P^j_i)$ must also appear in the previous level (also as a singleton), and every singleton participates in some simplification. Therefore, rooting the forest at its highest level, every tree has its leaves in level~0.

         We next argue that these trees induce a partition of $P$ into disjoint \emph{consecutive} subcurves. For every node $x$ in the forest, define $\mathsf{span}(x)$ to be the set of level-0 leaves in the subtree of $x$. We claim that $\mathsf{span}(x)$ always corresponds to a contiguous subcurve of $P$. This holds for level~0 trivially. For a node $x$ in level $j+1$, all of its children are nodes from level $j$ that lie in the same machine input $\I(P^{j+1}_i)$, and these inputs are ordered consistently with the order of $P$. Moreover, the children of $x$ are exactly the nodes represented by $x$ in a single $2\mu$-simplification applied to a maximal consecutive singleton subsequence. Hence $\mathsf{span}(x)$ is the union of consecutive child spans, and is therefore contiguous. Consequently, the root nodes $r_1,\dots,r_k$ have disjoint spans whose union is all of $P$, and these spans appear along $P$ in order. Let $p_1,\dots,p_k$ be the vertices corresponding to $r_1,\dots,r_k$, and let $Q_i$ be the subcurve of $P$ induced by $\mathsf{span}(r_i)$.

         In every level, we make a $2\mu$-simplification. Hence in total, the distance between $p_i$ and any leaf $w \in Q_i$ in its tree is at most $\|p_i-w\| \leq 2\mu \cdot (1/\rho)$, and the claim follows.   
    \end{enumerate}
\end{proof}

\begin{claim}\label{clm:simplification-many-rounds}
     After the merge round, $P'$ is a $(2/\rho+1)$-relaxed $\mu$-simplification of $P$. 
\end{claim}
\begin{proof}
We need to show that Properties  (\ref{property:SmallDistance}),(\ref{property:LongEdges}) of a $(2/\rho+1)$-relaxed $\mu$-simplification hold.

We start by showing Property~(\ref{property:SmallDistance}). Let $v \in P$, and let $Q_v$ be the curve containing $v$ in the partition guaranteed by \Cref{clm:app_simp_mpc_many_prop}(\ref{itm:app_simp_short_edge}). We split into cases:
\begin{enumerate}
    \item If $v \in Q_1$, we note that $p_1$ can only be deleted if $I_1 = p_1$ is a singleton, and assuming $\|p_1-p_2\| \leq \mu$. In this case, $\I_2 = (p_2,p_3)$ is the first pair, hence $p_2$ is not deleted in any of the cases, and by the triangle inequality $\|p_2-v\| \leq (2/\rho+1)\mu$. Next, we consider the deletions by iteration over the pairs.
    \item If we delete a vertex $p_i$ and $p_{i-1}$ survives, then $\|p_{i-1}-p_i\|\le \mu$, and by triangle inequality we have $\|p_{i-1}-v\| \leq (2/\rho+1)\mu$.
    \item Finally, we note that if we delete vertex $p_i$ and $p_{i-1}$, then necessarily $p_{i-2}$ survives and $\|p_{i-2}-p_i\|\le \mu$. Then by triangle inequality $\|p_{i-2}-v\| \leq (2/\rho+1)\mu$.
\end{enumerate}
    
    It remains to show that property (\ref{property:LongEdges}) holds. First we consider the special case where $I_1 = p_1$ is a singleton. Then $I_2 = (p_2,p_3)$ is a pair, and $p_2$ is undeleted by the process. If $\|p_1-p_2\| \leq \mu$, then $p_2$ becomes the first element in $P'$, and if $\|p_1-p_2\| > \mu$ then the condition holds for the two first elements of $P'$. Now consider a pair element $I_j = (p_i,p_{i+1})$. There are two cases:
    \begin{enumerate}
        \item $\I_{j+1}$ is a pair $(p_{i+2},p_{i+3})$, and $\|p_{i+1}-p_{i+2}\|\le\mu$. Then $p_{i+2}$ is deleted from $P'$, and by \Cref{clm:app_simp_mpc_many_prop}(\ref{itm:app_simp_long_edge}) we have that $\|p_{i+3}-p_{i+2}\| > 2\mu$. Hence by triangle inequality we have $\|p_{i+3}-p_{i+1}\|>\mu$.
        \item $\I_{j+1}$ is a singleton $p_{i+2}$, and $\I_{j+2}$ is a pair $(p_{i+3},p_{i+4})$. If $\|p_{i+1}-p_{i+2}\|\le\mu$ and $\|p_{i+1}-p_{i+3}\|\le\mu$ then both $p_{i+2},p_{i+3}$ are deleted, and as in the previous case, $\|p_{i+1}-p_{i+4}\|>\mu$ because $\|p_{i+3}-p_{i+4}\|> 2\mu$. If $\|p_{i+1}-p_{i+2}\|>\mu$ and $\|p_{i+2}-p_{i+3}\|\le\mu$, then by similar arguments $\|p_{i+2}-p_{i+4}\|>\mu$.
    \end{enumerate}
    Note that if $\I_{j+1}$ is a singleton $p_{i+2}$ and it is the last element in $\I$, then $p_{i+2}$ is deleted if $\|p_{i+1}-p_{i+2}\|\le\mu$, and $p_{i+1}$ will be the last vertex of $P'$. 
\end{proof}

\subsection{Simplification in the PRAM model}
\label{sec:simp_PRAM_appendix}
The algorithm and analysis of \Cref{thm:MPS-simp} can be translated in an essentially identical manner to the PRAM model, and we obtain as a corollary that an $O(\log n)$-relaxed $\mu$-simplification of a curve with $n$ vertices can be computed in $O(\log n)$ rounds using $O(n)$ processors.

We define the hierarchical partition almost exactly as before: in level zero we have the $n$ vertices of $P$, and for $j \geq 1$, the $j$-th level divides $P$ into $n/4^{j}$ consecutive subcurves, $P^j_1,\dots,P^j_{n/4^{j}}$ such that each subcurve is a concatenation of four consecutive subcurves from the previous level. The number of levels is therefore $O(\log{n})$, where the curve in the last level is $P$.

In step $j$, each processor is allocated a set $\I$ of at most four elements, and using $O(1)$ rounds computes on its own a $2\mu$-simplification for each maximal consecutive subsequence of singletons. The information that is then passed to the next step is as follows: if $\I$ has a single element $\I_1$, then pass it to the next level. Else, if $\I$ has at least two elements, then pass the first and the last elements to the next level, and in the case that both of them are singletons, also add one pair between them. We notice that at most three elements are passed to the next level from each processor. 

After the last step, we propagate all the information back to level zero, i.e., we get a sequence $\I$ of $O(n)$ elements (pairs and singletons), in which there are no two consecutive singletons. Then we run the merge procedure as described in \Cref{sec:simp_mpc}, which can be trivially performed in $O(1)$ rounds. Exactly as argued in the previous section, the resulting simplification is an $O(\log{n})$-relaxed $\mu$-simplification of $P$. 


\section{Realistic curves - missing proofs}
\subsection{Bounding the complexity of the matrix}\label{sec:packedness_missing_proofs}
As mentioned above, Driemel, Har-Peled and Wenk \cite{DHW12} show that if $P$ is $c$-packed, then the $\mu$-simplification $P^\mu$ is $6c$-packed. We provide a similar result for $\alpha$-relaxed $\mu$-simplifications. The proof follows the general approach of \cite{DHW12}, with minor modifications to account for the $\alpha$ factor in the distance. For completeness, we include a sketch of the proof.\omrit{proof can move to the appendix}
\omritin{TODO: reviewer 1: Claim 5.1: I think, here a visualization of the hippodromes would help a lot with the readability of the proof}
\begin{claim}\label{clm:simp_is_also_c_packed}
    If $P$ is $c$-packed, then any $\alpha$-relaxed $\mu$-simplification $P^\mu_\alpha$ of $P$ is $6\alpha c$-packed.
\end{claim}
\begin{proof}
    First, we claim that for any ball $B(p,r)$ it holds that ($\star$)$\|P\cap B(p,r+\alpha\mu)\|\ge \|P^\mu_\alpha\cap B(p,r)\|$.

    Let $S=(p_i,p_j)$ be an edge of $P^\mu_\alpha$ that intersects $B(p,r)$, and let $S'$ be this intersection. By the properties of $P^\mu_\alpha$, there is some index $i\le k\le j$ such that $B(p_i,\alpha\mu)$ contains $p_i,\dots,p_k$, and $B(p_j,\alpha\mu)$ contains $p_{k+1},\dots,p_j$. Thus, $P[i,j]=p_i,\dots,p_j$ is contained in the hippodrome of radius $\alpha\mu$ around $S$ (the Minkowski sum $S\oplus B(0,\alpha\mu)$). Consider the two hyperplanes passing through the endpoints of $S'$ and orthogonal to $S$. Observe that $P[i,j]$ must intersect both hyperplanes, and therefore the portion of $P[i,j]$ inside the hippodrome $S'\oplus B(0,\alpha\mu)$ is larger than $\|S'\|$. Because $S'\subseteq B(p,r)$, this hippodrome is contained in $B(p,r+\alpha\mu)$. By summing over all the edges, we get that $\|P\cap B(p,r+\alpha\mu)\|\ge \|P^\mu_\alpha\cap B(p,r)\|$.

    Next, assume by contradiction that $\|P^\mu_\alpha\cap B(p,r)\|> 6\alpha cr$ for some ball $B(p,r)$. 
    
    If $r\ge\alpha\mu$, then because $P$ is $c$-packed we have
    $\|P\cap B(p,r+\alpha\mu)\|\le c(r+\alpha\mu)\le 2cr$. By ($\star$) we get $\|P\cap B(p,r+\alpha\mu)\|\ge \|P^\mu_\alpha\cap B(p,r)\|> 6\alpha cr$, a contradiction.
    
    If $r<\alpha\mu$ then let $U$ denote the segments of $P^\mu_\alpha$ intersecting $B(p,r)$, and observe that each such segment can contribute at most $2r$ to the length of $P^\mu_\alpha$ inside $B(p,r)$, so $|U|> 6\alpha cr/2r=3\alpha c$. Since every edge of $P^\mu_\alpha$ has length of at least $\mu$, we have $\|P^\mu_\alpha\cap B(p,r+\alpha\mu)\|\ge \|U\cap B(p,r+\alpha\mu)\|\ge |U|\cdot\mu> 3\alpha\mu c$. By ($\star$) we get that $\|P\cap B(p,3\alpha\mu)\|\ge \|P^\mu_\alpha\cap B(p,2\alpha\mu)\|\ge \|P^\mu_\alpha\cap B(p,r+\alpha\mu)\|> 3\alpha\mu c$, a contradiction to the $c$-packedness of $P$.
\end{proof}

\begin{proof}[Proof of \Cref{lem:matrix-complexity}]
    Consider a ball $B$ of radius $(1+\frac{\eps}{2})\delta$ centered at some arbitrary point $w$ (see \Cref{fig:concentric_balls}).
    Denote by $V$ the set of vertices of $P^\mu$ that are contained in $B$, and consider the ball $B'$ of radius $(1+\frac{\eps}{2}+\frac{\eps}{4\alpha})\delta$ centered at $w$. Since the length of each edge of $P^\mu$ is at least $\frac{\eps\delta}{4\alpha}$, and the radius of $B'$ is larger by $\frac{\eps\delta}{4\alpha}$ than the radius of $B$, we get that each vertex in $V$ is adding at least $\frac{\eps\delta}{4\alpha}$ to the length of $P^\mu$ inside $B'$. So the total length of $P^\mu$ inside $B'$ is at least $|V|\frac{\eps\delta}{4\alpha}$. Because $P^\mu$ is $6\alpha c$-packed (by \Cref{clm:simp_is_also_c_packed}), the length of $P^\mu$ inside $B'$ is at most $6\alpha c (1+\frac{\eps}{2}+\frac{\eps}{4\alpha})\delta$, we get that $|V|\frac{\eps\delta}{4\alpha}\le 6\alpha c (1+\frac{\eps}{2}+\frac{\eps}{4\alpha})\delta$, and $|V|=O(\frac{\alpha^2c}{\eps})$.
    We conclude that any ball $B$ of radius $(1+\frac{\eps}{2})\delta$ contains $O(\frac{\alpha^2c}{\eps})$ vertices of $P^\mu$. The 1-entries of row $i$ in $M$ correspond to vertices of $P^\mu$ that are contained in the ball of radius $(1+\frac{\eps}{2})\delta$ centered at the $i$-th vertex of $Q^\mu$, and therefore we get that there are $O(\frac{\alpha^2c}{\eps})$ such entries.
    By symmetric arguments, the claim also holds for the columns of $M$.
\begin{figure}[h!]
    \centering
    \includegraphics[scale=1]{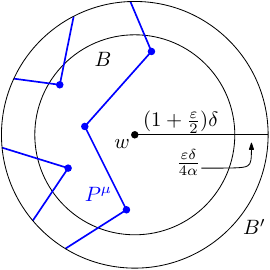}
    \caption{The concentric balls $B$ and $B'$. The edges of $P^\mu$ inside $B$ are in blue.}
    \label{fig:concentric_balls}
\end{figure}
\end{proof}

\subsection{The decision problem in the PRAM model}\label{sec:decision_PRAM}
In the PRAM model, we compute the simplification via \Cref{thm:MPS-simp}, and apply the procedure relating to PRAM in \Cref{thm:relevant-submatrices} with a constant $\sigma > 0$ to find the relevant submatrices. Each such relevant submatrix $M'$ has $n^\sigma$ rows and columns, and $O(\frac{\alpha^2c}{\eps})$ 1-entries on its boundary. We apply in parallel for each such matrix $M'$  and each pair of 1-entries on its boundary the directed reachability algorithm mentioned above, and store the set of 1-entries on its boundary that are reachable from it. We are left with a graph of size $\Tilde{O}(\frac{c}{\eps \sigma^2} n^{1-\sigma})$. Like in MPC, we recurse until the graph is of size $O(n^{\sigma})$; this recursion has depth $O(1/\sigma)$. Recall that given a directed graph $G = (V,E)$, it is possible to determine in PRAM for all pairs of vertices $u,v$ whether $u$ can reach $v$ in $O(\log{|V|})$ rounds, using $O(|V|^3)$ processors and $O(|V|^3\log{|V|})$ total work (see \cite{J92}). Finally, we apply the directed reachability algorithm mentioned above to determine whether $M[1,1]$ can reach $M[n,n]$. In total, we use $O(\sigma \log{n}+1/\sigma^2)$ rounds. 

\subsection{Optimization}\label{sec:optimization_missing_proofs}

\begin{lemma}
\label{lem:distance_magnitude_one_d}
    In the MPC model, given a set of points $X = \{x_1,\dots,x_n\} \in \reals$, we can output a set $D \subseteq \reals$ of size $|D| = \Tilde{O}(n)$ such that for any two points $x,y \in X$ there exists a value $r \in D$ for which $r/2 \leq |x-y| \leq r$. The algorithm uses $O(1)$ rounds and $\Tilde{O}(n)$ total space.
\end{lemma}
\begin{proof}
    For simplicity, we assume w.l.o.g. that $n = 2^k$ for some integer $k$, and assume $x_1,\dots,x_n$ are ordered in ascending order. For $1 \leq i \leq \log{n}$, let $S_i = \{n/2^i,2n/2^i,3n/2^i,\dots,(2^i-1)n/2^i\}$ be a set of indices which partitions $[n]$ into segments of size $n/2^i$. For an index $j \in [n]$, let 
    $C(j) = \bigcup_{1 \leq i \leq \log{n}} \{\max_{s \in S_i, s \leq j} s\} \cup \{\min_{s \in S_i, s \geq j} s\}$, be the closest indices to $j$ in each of $S_i$ from both left and right. In particular, we note that $|C(j)| = O(\log{n})$, since at most two indices are added for each $i \in [\log{n}]$. 
    
    Next, we show that for any $j_1,j_2 \in [n]$ where $j_1 < j_2$, there exists a value $j' \in C(j_1) \cap C(j_2)$ such that $j_1 \leq j' \leq j_2$. Indeed, this holds as for the first $i$ such that $S_i$ separates $j_1,j_2$, i.e. there exists $j' \in S_i$  $j_1 \leq j' \leq j_2$. Moreover, we note that $j' \in S_i$  is contained in both $C(j_1),C(j_2)$: assume otherwise, therefore there must be two consecutive indices in $j'_1,j'_2 \in S_i$ such that $j_1 \leq j'_1 < j'_2 \leq j_2$. But by construction of $S_i$, any two consecutive indices hold that either $j'_1 \in S_{i-1}$ or $j'_2 \in S_{i-1}$, contradicting the minimality of $i$. 
    
    Our algorithm is defined as follows: First, we locally compute $C(j)$ for all $j \in [n]$. This can be done since it only requires knowledge of the value $n$. Next, for each $x_j$, we simultaneously compute the distances $D_j = \{|x_j-x_{j'}| \mid j' \in C(j)\}$, and set $D = \bigcup_{j \in [n]} D_j$. This can be trivially done in $O(1)$ rounds, using $O(n\log{n})$ total space. 
    
    Finally, we show that for any $x,y \in X$, there exists $r \in D$ such that $r/2 \leq |x-y| \leq r$. Recall that we showed that we have a point $z$ such that $x \leq z \leq y$ and such that both $|x-z|,|y-z| \in D$. Assuming w.l.o.g. that $|x-z| \leq |y-z|$, we have that $|x-y|/2 \leq |y-z| \leq |x-y|$, and the claim follows. 
\end{proof}

\begin{lemma}
\label{lem:distance_magnitude_any_d}
    In the MPC model, given a set of points $X = \{x_1,\dots,x_n\} \subseteq \reals^d$ and constant value $\eps > 0$, we can output a set $D \subseteq \reals$ of size $|D| = \Tilde{O}(nd)$ such that for any two points $x,y \in X$ there exists a value $r \in D$ for which $r \leq \|x-y\|_2 \leq (1+\eps)r$. The algorithm uses $O(1)$ rounds and total space of $\Tilde{O}(nd)$.
\end{lemma}
\begin{proof}
    Let $T = \lceil \log_{(1+\eps)}(2\sqrt{d}) \rceil$. Simultaneously for each coordinate $i \in [d]$, we consider the set of points $X_i = \{x^i_1,\dots,x^i_n\} \in \reals$, where $x_j^i$ is the $i$-th coordinate of $x_j$, and compute a set $D_i$ according to \Cref{lem:distance_magnitude_one_d}. Following this, we set 
    $D = \{(1+\eps)^\ell \cdot r \mid \ell \in [-T,T], r \in \cup_{i \in [d]} D_i\}$.

    We show that $D$ satisfies the property we require. For any two points $x,y$, and coordinate $i$, we have a value $r_i \in D_i$ such that $r_i/2 \leq |x^i-y^i| \leq r_i$. In particular, for some $i \in [d]$ we have a value $r \in D_i$  such that 
    $r/2 \leq \|x-y\|_{\infty} \leq r$. Moreover, for any $\ell \in [-T,T]$, we have that $(1+\eps)^\ell r \in D$.
    
    We recall the well-known inequality that for any point $x \in \reals^d$, it holds that $\|x\|_\infty \leq \|x\|_2 \leq \sqrt{d} \cdot \|x\|_\infty$. Combining this with the statement above, we conclude that there exists an $r' \in D$ such that 
    $r' \leq \|x-y\|_2 \leq (1+\eps)r'$, and the claim follows.
\end{proof}

To prove \Cref{thm:main-c-packed}, we prove the two following theorems separately.
\begin{theorem}
\label{thm:main_finding_cpack_MPC}
    Let $\rho < 1,\eps > 0$, be constants, and assume that $\frac{\log{c}}{\log{n}}$ is smaller than some constant. Given two $c$-packed curves of length $n$, a $(1+\eps)$-approximation of their DFD can be found in the MPC model in $O(1/\rho^3)$ rounds, using $n^{\rho}$ space per machine, and using $\Tilde{O}(n^{1+\rho})$ total space. 
\end{theorem}
\begin{proof}
     We reduce to the DFD decision problem in the following simple manner: let $x_1,\dots,x_n  \in \reals^d$, and $y_1,\dots,y_n \in \reals^d$ be the points of the two curves respectively. Set $X = \{x_1,\dots,x_n,y_1,\dots,y_n\}$ and compute a set $D$ according to \Cref{lem:distance_magnitude_any_d}, setting the lemma's parameter $\eps = \eps$. Recall that $D$ is defined such that we are guaranteed that for every two points $x,y$ of the two curves respectively, there exists a value $r \in D$ for which $r \leq \|x-y\|_2 \leq (1+\eps)r$. In particular, we are guaranteed that $D$ contains a $(1+\eps)$-approximation of the DFD.
     
     We sort the values $D = \{r_1,\dots,r_{|D|}\}$, and run a multi-pivot search to find an approximation of the DFD: Recall that the decision algorithm of \Cref{thm:realistic_mpc_descision} allows us to determine whether a given value $\delta$ is a $(1+\eps)$-approximation of the DFD of the two curves, or whether it is too small or too large. We let $r'_1,\dots,r'_{n^\rho} \in D$ be the $n^\rho$-quantiles of $D$, and run the algorithm $n^\rho$ times in parallel, where for each we set $\delta = r'_1,\dots,r'_{n^\rho}$ respectively, and set our approximation parameter as $\eps$. After identifying the relevant segment in $D$ which contains a $(1+\eps)$-approximation of the DFD, we recursively run on that segment. We note that the depth of this recursion is $O(1/\rho)$, i.e. is of constant depth.
\end{proof}

\begin{theorem}
\label{thm:main_finding_cpack_PRAM}
    For any constant $\eps > 0$ and value $\Theta(\frac{\log{c}}{\log{n}}) \leq \sigma < 1$, given two $c$-packed curves of length $n$, a $(1+\eps)$-approximation of their DFD can be found in the PRAM model in $O(\log{n}+1/\sigma^3)$ rounds, and $\Tilde{O}(n^{1+\sigma})$ total work.
    
\end{theorem}
\begin{proof}
In the PRAM model, we again adopt an identical approach: we can compute a set $D$ as in \Cref{lem:distance_magnitude_any_d} using $\Tilde{O}(n)$ processors and $O(\log{n})$ rounds, by first constructing a relaxed set for one dimension (which only requires sorting and other trivial operations), and then computing $D$ for general dimension (which uses only trivial operations). Finally, we perform a multi-pivot search, and conclude that the DFD can be approximated in $O(\log{n}+1/\sigma^3)$ rounds, for any $\sigma > 0$.
\end{proof}

\old{
\section{Conclusion and open problems}
\label{sec:conclusions}
In this paper, we provide the first parallel algorithms for the discrete \frechet\ distance in both the MPC and PRAM models. In addition, we present more efficient parallel algorithms for $c$-packed curves, in both models. Our algorithms for $c$-packed curves also require the implementation of a simplification algorithm in parallel, which may be of independent interest. All our algorithms are close to optimal, given the SETH-based lower bounds in $\cite{Bringmann14}$. 
Clearly, as this is the first paper on this problem, many questions remain open. We state two of them here, which we believe are the natural extension of our work. 
The first open question is computing the continuous \frechet\ distance in parallel. In our paper we focus on the discrete version, which is also slightly easier to compute in the sequential model. The main challenge in computing the continuous distance in parallel is the optimization procedure, because the number of ``critical events'' becomes $n^3$ instead of $n^2$. The optimization procedure in \cite{AG95} uses parametric search which is rather complex, however, in \cite{HR11} the optimization procedure is simpler and based on sampling critical values, and might be easier to implement in parallel.\omrit{not sure we want this last sentence} The second open question is computing the (discrete) \frechet\ distance between two curves such that only one of them is $c$-packed. This problem was considered in the sequential model by Gudmundsson, Mai and Wong~\cite{GMW24} for the continuous distance, and by van der Hoog, Rotenberg and Wong~\cite{HRW24} for the discrete distance. In both cases there are near-linear time algorithms, and it will be interesting to see if we can achieve efficient parallel algorithms in this case.
}
\end{document}

%% file: mymacros-disc.tex
\usepackage{graphicx}
\usepackage{xcolor}
\usepackage[margin=1in,footskip=0.25in]{geometry}
\definecolor{forestgreen}{rgb}{0.13, 0.55, 0.13}

\graphicspath{{./figures/}}

\def\I{\mathcal{I}}

\def\reals{{\mathbb R}}

\def\eps{{\varepsilon}}

\newcommand{\frechet}{Fr\'echet}
\newcommand{\dfd}{d_{dF}}

\renewcommand{\arraystretch}{1.5}

\def\dfn#1{\emph{\textbf{#1}}}

\date{}

\newcommand{\old}[1]{{{}}}
\newcommand{\temp}[1]{{{}}}
\newcommand{\fullversion}[1]{{{}}}

\newcommand{\oldparagraph}[1]{\vspace{5pt}\noindent\textbf{#1}}

\newcommand{\omrit}[1]{}
\newcommand{\omritin}[1]{}
\newcommand{\orr}[1]{}
\newcommand{\orrin}[1]{}
\newcommand{\arnold}[1]{}
\newcommand{\arnoldin}[1]{}
\newcommand{\leonid}[1]{}
\newcommand{\leonidin}[1]{}

%% file: sec_models.tex
\subsection*{Models of computation}

\oldparagraph{PRAM:} The PRAM model is perhaps the most prominent model in the study of parallel algorithms for geometric problems (see \cite{GO04} Chapter 46). It is a shared memory model where the processors operate synchronously. In each step, processors can read and write to the shared memory (and perform local computations), and thus all communication is done via the shared memory. In each node, we can perform a constant number of operations, given an input word. 

An algorithm can be represented as a rooted DAG, with constant incoming degree, where the first layer represents the input (i.e. we have a set of nodes without an incoming edge, each representing a word of input), and each node in a layer following this represents an operation applied to the results of prior operations or to input words.

In the analysis of a PRAM algorithm, three complexity parameters are considered: (i) the number of rounds, which is the height of the DAG, (ii) the number of processors, which is the maximum number of nodes in any layer of the DAG but the first, and (iii) total work done, which is the number of nodes after the first layer.

There are several well studied variants of PRAM (see \cite{J92} chapter 1 for an extensive introduction). For the purposes of this work, we focus solely on the ``common'' CRCW PRAM, where several processors may read the same memory cell in a single round, and may write into the same cell as long as they attempt to write the same value.

\oldparagraph{MPC:} \arnold{AF: perhaps mention that it represents some real-life models by name, like Hadoop and MapReduce}
The Massively Parallel Computation (MPC) model is another extensively studied model of parallel computation. It is a synchronous distributed model in which we have a set of $M$ machines, each with a unique identifier in $[1,M]$, and each with local space $S$. Initially, an input of size $N$ memory words is partitioned between the machines in an arbitrary manner. In each round, each machine may send a message to each other machine, where the only restriction on the size of the messages is that the total number of memory words each machine sends or receives is at most $S$ per round. After each machine receives its messages, it performs a local computation step, in which it can compute a new state encoded using $S$ words of memory, based on its current state and the received messages of the latest round. The protocol ends when all machines terminate with some output. The round complexity of a protocol is the worst case number of rounds after which all machines terminate. The total space is $M \cdot S$, and we assume $M \cdot S \geq N$.

\oldparagraph{Useful procedures.}
We rely on the following lemma from prior works, which shows that sorting and aggregate functions (i.e. associative and commutative) can be computed in MPC in constant time.

\begin{lemma}[Aggregation and Sorting, Informal; see e.g. \cite{GSZ11}]
    \label{lem:aggregate}
    Assume each machine has space $N^\rho$, for some integer $N$ and value $\rho > 0$. Given a set of $N$ integer elements, which is explicitly stored on the machines, and a partition $\mathcal{P} = \{P_1,\dots,P_r\}$ of the elements, known to all machines, we can compute an aggregate function (e.g. sum, $\min$, $\max$) on each of $P_1,\dots,P_r$ simultaneously in $O(1/\rho)$ rounds. Similarly, the elements of each of $P_1,\dots,P_r$ can be sorted\footnote{In the sorting problem, each element receives as output its rank in the sorted multi-set, and the elements are reshuffled so that for any two elements $x,y$ where $x < y$, the id of the machine storing $x$ is no larger than the id of the machine storing $y$.} simultaneously in $O(1/\rho)$ rounds.
\end{lemma}

In several of our claims analyzing the round and space complexity of our MPC (resp. PRAM) algorithms, we make an implicit re-parametrization argument where we prove the claim for total space of $n^{k+O(\rho)}$ (resp. number of processors $n^{k+O(\sigma)}$) for some $k$. By choosing $\rho$ (resp. $\sigma$) to be a constant smaller and re-parameterizing by this new value, we can remove the constant factors in the space complexity, at the cost of a larger constant factor in the round complexity. For brevity, we omit these details from the proofs. 

%% file: sec_trivialAlgs.tex
\section{``Trivial'' algorithms in the PRAM model}\label{subsec:trivial}
In this section we present a few ``trivial'' algorithms for computing the discrete \frechet\ distance in the PRAM model, to serve as a baseline. Each of these algorithms has some advantage over the others, either by its simplicity or by one of the three parameters mentioned in the introduction.

The observation below follows from the dynamic programming algorithm for computing the DFD of two curves.

\begin{observation}
\label{obs:trivial_alg_n_time}
    Given two curves of length $n$, their DFD can be computed in the PRAM model in $2n-1$ rounds, using $O(n^2)$ total work.
\end{observation}
\begin{proof}
    Consider $2n-1$ ``diagonals'' of the matrix $M_\delta$, as depicted in the figure below: In the $i$-th round, $1\le i\le 2n-1$, at most $n$ processors will compute the reachability values of the cells in the $i$-th diagonal (the cells $(x,y)$ such that $x+y=i+1$).
    \begin{figure}[h]
        \centering
        \includegraphics[page=3, scale=1]{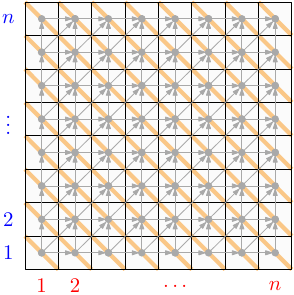}
    \end{figure}
\end{proof}
Alternatively, another simple solution is to use the best known bound for directed reachability in the PRAM model:

\begin{observation}
\label{obs:trivial_alg_polylog_time}
    Given two curves of length $n$, their DFD can be computed in the PRAM model in $O(\log{n})$ rounds, using $O(n^6\log{n})$ total work.
\end{observation}
\begin{proof}
    First, the processors check that $M_\delta[1,1] = M_\delta[n,n] = 1$, and output zero otherwise. Using a very simple algorithm for transitive closure (see \cite{J92}), 
    given a directed graph $G = (V,E)$ and two vertices $u,v \in V$, it is possible to determine in PRAM whether $u$ can reach $v$ in $O(\log{|V|})$ rounds, using $O(|V|^3)$ processors and $O(|V|^3\log{|V|})$ total work. Running this procedure on the directed graph corresponding to $M_\delta$, we can determine whether $M_\delta[1,1]$ can reach $M_\delta[n,n]$ in the graph described above using the same number of rounds, processors, and total work, where $|V| = O(n^2)$. 
\end{proof}

 Finally, we note that we can trivially compute the DFD in $O(1)$ rounds, but with $2^{O(n)}$ total work.

\begin{observation}
\label{obs:trivial_alg_const_time}
    Given two curves of length $n$, their DFD can be computed in the PRAM model in $O(1)$ rounds, using $2^{O(n)}$ total work. 
\end{observation}
\begin{proof}
    The processors simultaneously check all possible directed paths between $M_\delta[1,1]$ and $M_\delta[n,n]$. There are $2^{O(n)}$ such paths. Each path can be locally verified by a set of at most $2n$ processors in one round, by checking if all cells of that path are 1-entries, and storing this information in a memory cell designated for that path. Next, we check whether at least one path exists and store this in the output cell.
\end{proof}

%% file: sec_computingDFD.tex
\section{Computing the discrete \frechet\ distance in parallel}\label{sec:discrete-frechet}
In this section we present parallel algorithms for computing the discrete \frechet\ distance between two polygonal curves $A$ and $B$, each consisting of $n$ vertices in $\reals^d$. 
We first consider the decision version of the problem, where given a threshold $\delta$ the algorithm should decide whether $\dfd(A,B)\le\delta$. 

Given threshold value $\delta$ and two polygonal curves $A$ and $B$, each with $n$ vertices in $\reals^d$, consider the corresponding binary matrix $M_\delta$. 
The main idea is to start by dividing the matrix $M_\delta$ into thin ``strips'' (sequences of consecutive columns), and then in each round merge pairs of consecutive strips. The size of a single strip can be $\Theta(n^2)$, however, we observe that for the merge procedure it is enough to maintain the reachability information only for the ``boundary'' columns of the strip (i.e. the leftmost and rightmost columns). We then show that the reachability information can be maintained by storing only a constant number of values per boundary cell.

\begin{figure}[h]
    \centering
    \includegraphics[page=4, scale=1]{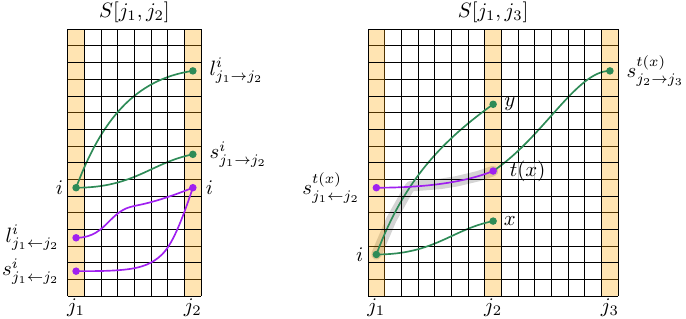}
    \caption{Left: the reachability information for the boundary cells in row $i$ of the strip $S[j_1,j_2]$. Right: merging two strips.}
    \label{fig:crossing_paths}
\end{figure}

A \dfn{strip} $S[j_1,j_2]$ of $M_\delta$ is the submatrix of $M_\delta$ consisting of the columns $j_1,j_1+1,\dots,j_2$ of $M_\delta$.
Given a strip $S[j_1,j_2]$, suppose that we have already computed for every $i\in [n]$ the following values (see \Cref{fig:crossing_paths}):
\begin{itemize}
    \item The smallest index $\overrightarrow{s}=s^i_{j_1\to j_2}$ (resp. largest index $\overrightarrow{l}=l^i_{j_1\to j_2}$) such that there is a monotone path from $M_\delta[j_1,i]$ to $M_\delta[j_2,\overrightarrow{s}]$ (resp. $M_\delta[j_2,\overrightarrow{l}]$).
    \item The smallest index $\overleftarrow{s}=s^i_{j_1\leftarrow j_2}$ (resp. largest index $\overleftarrow{l}=l^i_{j_1\leftarrow j_2}$) such that there is a monotone path from $M_\delta[j_1,\overleftarrow{s}]$ (resp. $M_\delta[j_1,\overleftarrow{l}]$) to $M_\delta[j_2,i]$.
\end{itemize}
We call this set of values the \dfn{reachability information} of $S[j_1,j_2]$. 

To merge the reachability information from two strips, we need a few facts and observations. The first is a known result for showing how to compute for each cell $i$ in a binary array $T$ the next index $j \geq i$ such that $T[j] = 1$. \orr{and the second?}

\begin{lemma}[Find-Smaller problem \cite{BV94}]\label{obs:computing_t}
    Let $T$ be an array of length $n$ such that $T[i]\in \{0,1\}$ for every $i\in [n]$. In the PRAM model, we can compute in $O(1)$ rounds and using $\Tilde{O}(n)$ total work, an array $T'$ such that $T'[i]$ is the smallest index $j\ge i$ for which $T[j]=1$.
\end{lemma}

\begin{observation}[crossing paths]\label{obs:crossing_paths}
    For any four indices $i'< i\le k< k'$, if there is a monotone path from $M_\delta[j_1,i]$ to $M_\delta[j_2,k]$ and a monotone path from $M_\delta[j_1,i']$ to $M_\delta[j_2,k']$, then there are also monotone paths from $M_\delta[j_1,i]$ to $M_\delta[j_2,k']$ and from $M_\delta[j_1,i']$ to $M_\delta[j_2,k]$.
\end{observation}
\begin{proof}
    Because they are monotone, the paths from $M_\delta[j_1,i]$ to $M_\delta[j_2,k]$ and from $M_\delta[j_1,i']$ to $M_\delta[j_2,k']$ must intersect at some entry $M_\delta[j,x]$ for $j_1\le j\le j_2$ and $i\le x\le k$.
    Then, there are monotone paths: from $M_\delta[j_1,i]$ to $M_\delta[j,x]$ and then to $M_\delta[j_2,k']$, and from $M_\delta[j_1,i']$ to $M_\delta[j,x]$ and then to $M_\delta[j_2,k]$.
\end{proof}

\begin{claim}
\label{clm:merge_strips}
    Given two strips $S[j_1,j_2]$ and $S[j_2,j_3]$ (with one column in common), and the reachability information of both $S[j_1,j_2]$ and $S[j_2,j_3]$, the reachability information of $S[j_1,j_3]$ can be computed in the PRAM model using $\Tilde{O}(n)$ total work in $O(1)$ rounds.
\end{claim}
\begin{proof}
    First, for every index $i\in [n]$ we compute the smallest index $t(i)\ge i$ such that both $s^{t(i)}_{j_2\to j_3}$ and $s^{t(i)}_{j_1\leftarrow j_2}$ exist. By \Cref{obs:computing_t}, this can be done using $\Tilde{O}(n)$ processors in $O(1)$ rounds, on an array $T$ such that $T[i]=1$ if both $s^{i}_{j_2\to j_3}$ and $s^{i}_{j_1\leftarrow j_2}$ exist, and $T[i]=0$ otherwise.
    Next, we show that for a fixed index $i\in [n]$, computing $s^i_{j_1\to j_3}$ can be done using a constant number of operations. The other three values, $l^i_{j_1\to j_3}$, $s^i_{j_1\leftarrow j_3}$, and $l^i_{j_1\leftarrow j_3}$ can be computed similarly.
    
    Denote $x=s^i_{j_1\to j_2}$ and $y=l^i_{j_1\to j_2}$. If $x,y$ do not exist, then clearly $s^i_{j_1\to j_3}$ does not exist. We therefore assume that $x,y$ exist (it is possible that $x=y$). Consider some $x\le k\le y$ such that $i'=s^{k}_{j_1\leftarrow j_2}$ exists, i.e. there is a monotone path from $M_\delta[j_1,i']$ to $M_\delta[j_2,k]$. By \Cref{obs:crossing_paths}, there must be a monotone path from $M_\delta[j_1,i]$ to $M_\delta[j_2,k]$.
    
    Now consider the value $t(x)$ that was computed previously  (see \Cref{fig:crossing_paths}). If $t(x)>y$, or $t(x)$ does not exist, then clearly there is no monotone path from $M_\delta[j_1,i]$ to any entry in the column $j_3$, because such a monotone path would have to include an entry $M_\delta[j_2, k]$ for some $x\le k\le y$. Otherwise, if $t(x)\le y$, there must be a monotone path from $M_\delta[j_1,i]$ to $M_\delta[j_3,s^{t(x)}_{j_2\to j_3}]$, and by the minimality of $t(x)$ we have $s^i_{j_1\to j_3}=s^{t(x)}_{j_2\to j_3}$.
\end{proof}

As there are $O(\log{n})$ rounds of merging consecutive strips, we obtain the following theorem.
\begin{theorem}
\label{thm:PRAM_main}
    Given two curves of length $n$, deciding whether the DFD is larger than some value $\delta$ can be performed in the PRAM model in $O(\log{n})$ rounds, with $\Tilde{O}(n^{2})$ total work.
\end{theorem}

In \Cref{apx:discrete-frechet}, we show that the algorithm from \Cref{thm:PRAM_main} immediately implies an algorithm for the MPC model with $O(n^\rho)$ space per machine, round complexity $O(\log{n}/\rho)$ and total space $\Tilde{O}(n^2)$ (see \Cref{thm:MPC_main}).
Then we show how to solve the optimization problem, proving \Cref{thm:exact-dfd}. Finally, we discuss the adaptation of this algorithm (and the algorithm of \cite{AG95}) to the continuous \frechet\ distance.

%% file: sec_computingDFD-apx.tex
\section{Computing the DFD in parallel - missing proofs}\label{apx:discrete-frechet}

\subsection{Decision in the MPC model}
In this section, we show that the algorithm from \Cref{thm:PRAM_main} immediately implies an algorithm for the MPC model with $O(n^\rho)$ space per machine, round complexity $O(\log{n}/\rho)$ and total space $\Tilde{O}(n^2)$.

\begin{theorem}
\label{thm:MPC_main}
    Given two curves of length $n$, the decision version of their DFD can be computed in the MPC model in $O(\log{n}/\rho)$ rounds, with local space $\Theta(n^{\rho})$ per machine and with total space $\widetilde{\Theta}(n^2)$, where $\rho > 0$ is any constant.
\end{theorem}
\begin{proof}
    Following a very similar argument of the algorithm of \Cref{thm:PRAM_main}, it suffices to show that given the reachability information of two strips $S[j_1,j_2]$ and $S[j_2,j_3]$, the reachability information of $S[j_1,j_3]$ can be computed in $O(1/\rho)$ rounds and $\Tilde{O}(n)$ total space. Given this sub-procedure, we perform $O(\log{n})$ iterations as in \Cref{thm:PRAM_main}, where in each such iteration we ``merge'' the reachability information of pairs of strips, and the claim follows.
    
    As in \Cref{clm:merge_strips}, computing $S[j_1,j_3]$ from $S[j_1,j_2]$ and $S[j_2,j_3]$ can be reduced to solving the problem of \Cref{obs:computing_t}: given an array $T$ of length $|T| = n$ such that $T[i]\in \{0,1\}$ for every $i\in [n]$, we need to output an array $T'$ such that $T'[i]$ is the smallest index $j\ge i$ for which $T[j]=1$. 
    This problem can be easily solved using sorting in $O(1/\rho)$ rounds: for each cell in $T$, we additionally store its index in $T$. Then, we use the sorting procedure of \Cref{lem:aggregate} to sort by cell value, and then by index. We thus obtain a consecutive sub-array of all the indices for which $T[i] = 1$, sorted by their original index in $T$. Using this sub-array, we can trivially construct the array $T'$ in $O(1/\rho)$ rounds.   
\end{proof}

\subsection{Optimization}

\thmExactDFD*
\begin{proof}
     In both models, we reduce to the decision version of the problem. In MPC, we compute and sort the $n^2$ distances between any two points of the two curves. Denote this set $D = \{r_1,\dots,r_{n^2}\}$, and note that the DFD value is contained in $D$. We then apply a multi-pivot search\footnote{A multi-pivot search is similar to a binary search, but instead of a single pivot we have many pivots, and recurse on the relevant segment.} to find the value of the DFD in the following manner: The decision algorithm of \Cref{thm:MPC_main} for MPC allows us to determine whether for a given value $\delta$ bounds from above the DFD of the two curves. We let $r'_1,\dots,r'_{n^\rho} \in D$ be the $n^\rho$-quantiles of $D$.\footnote{The $k$-quantiles of a sorted list of integers $A$ is a set of $k-1$ elements that partition the list into $k$ blocks of almost equal size.} We run the decision algorithm  $n^\rho$ times in parallel, where for each we set $\delta = r'_1,\dots,r'_{n^\rho}$ respectively. We then identify the relevant segment in $D$ which contains the exact value of the DFD, and we recursively run on that segment. We note that the depth of this recursion is $O(1/\rho)$, and at the end of the search we are guaranteed to find the value of the DFD in $D$. Similarly in PRAM, we compute $D$ (without sorting), and perform a multi-pivot search, where in each step we choose a set of random elements $r'_1,\dots,r'_{n^\sigma}$ from $D$, and run the decision algorithm $n^\sigma$ times in parallel, where for each we set $\delta = r'_1,\dots,r'_{n^\sigma}$ respectively. Then, using a standard procedure\footnote{For an $O(k)$-sized array it is easy to find in $O(1)$ rounds and $k^2$ processors the minimum and maximum elements. We can assume that $\sigma < 1/2$ w.l.o.g., with a cost of a constant in round complexity. Therefore, we can mark all pivots larger than the DFD, and find the minimum pivot larger than the DFD. Similarly, we can find the largest pivot smaller than the DFD. The DFD is contained in the segment between these two pivots.} we find in $O(1)$ rounds the two pivots between which the DFD is contained, and recurse, where the depth is $O(1/\sigma)$.
\end{proof}

\subsection{The continuous \frechet\ distance}\label{sec:continuous}
As mentioned in the introduction, Alt and Godau~\cite{AG95} presented an algorithm for computing the \frechet\ distance between closed curves, which uses a similar divide-and-conquer approach for the decision procedure: they split the free space diagram into smaller diagrams, on which they recursively compute the reachability information for the entire boundary. They do this by partitioning the boundary into intervals, and computing the lowest and highest point on the opposite side of the boundary which is reachable from each interval.

In the case of computing the classic \frechet\ distance, a slightly simpler approach can be taken by splitting the diagram into strips, as we do for the discrete case. Then, we only need to recursively compute the reachability information for the right and left boundary of each strip, by dividing the boundary into intervals as in~\cite{AG95}. The parallel algorithms in the PRAM and MPC model are then similar to the algorithms that we describe in this section for the discrete version.